\documentclass[letterpaper,11pt]{article}
\usepackage{amsfonts}
\usepackage{amssymb, amsmath, amsthm,  graphicx}
\usepackage{slashed}
\usepackage{array}
\usepackage{natbib}
\usepackage[margin=1in]{geometry}
\usepackage[bottom]{footmisc}
\usepackage{commath}
\usepackage{url,multirow,booktabs,enumerate}
\usepackage{rotating,lscape}
\usepackage[table]{xcolor}

\DeclareMathOperator*{\argmin}{argmin}
\newtheorem{thm}{Theorem}

\newtheorem{lemma}{Lemma}

\newtheorem{assum}{Assumption}
\newtheorem{definition}{Definition}

	\newcolumntype{H}{>{\setbox0=\hbox\bgroup}c<{\egroup}@{}}
	\newcommand{\rotpar}[2][2.2cm]{\rotatebox{70}{\parbox{#1}{\raggedright #2}}}

\begin{document}

\title{Sharp Threshold Detection Based on Sup-norm Error rates in High-dimensional Models}
\author{\textsc{Laurent Callot\thanks{VU University Amsterdam, Department of Econometrics and Operations Research, CREATES - Center for Research in Econometric Analysis of Time Series (DNRF78) funded by the Danish National Research Foundation, and the Tinbergen Institute.}}\and \textsc{Mehmet Caner\thanks{%
North Carolina State University, Department of Economics, 4168 Nelson Hall,
Raleigh, NC 27695. Email: mcaner@ncsu.edu}} \and \textsc{Anders Bredahl Kock\thanks{%
Aarhus  University, Department of Economics and Business, and CREATES - Center for Research in Econometric Analysis of Time Series (DNRF78), funded by the Danish National Research Foundation. Fuglesangs Alle 4, 8210, Aarhus V Denmark. Email: akock@creates.au.dk
}} \and and \textsc{Juan Andres Riquelme\thanks{%
North Carolina State University, Department of Economics,}}}
\date{\today}

\maketitle

\begin{abstract}
We propose a new estimator, the thresholded scaled Lasso, in high dimensional threshold regressions. First, we establish an upper bound on the $\ell_\infty$ estimation error of the scaled Lasso estimator of \cite{lee2012lasso}. This is a non-trivial task as the literature on high-dimensional models has focused almost exclusively on $\ell_1$ and $\ell_2$ estimation errors. We show that this sup-norm bound can be used to distinguish between zero and non-zero coefficients at a much finer scale than would have been possible using classical oracle inequalities. Thus, our sup-norm bound is tailored to consistent variable selection via thresholding. 

Our simulations show that thresholding the scaled Lasso yields substantial improvements in terms of variable selection. Finally, we use our estimator to shed further empirical light on the long running debate on the relationship between the level of debt (public and private) and GDP growth.   

\vspace{0.1in}


\end{abstract}

\thispagestyle{empty}\setcounter{page}{0}\newpage


\section{Introduction}
Threshold models have been heavily studied and used in the past twenty years or so. In econometrics the seminal articles by \cite{hansen1996} and \cite{hansen2000} showed that least squares estimation of threshold models is possible and feasible. 
These papers show how to test for the presence of a threshold and how to estimate the remaining parameters by least squares. Later, \cite{caner2004hansen} provided instrumental variable estimation of the threshold. These authors derived the limits for the threshold parameter in the reduced form as well as structural equations. 

There have been many applications of threshold models in cross-section data. One of the most recent ones is the analysis of the public debt to GDP ratio in a threshold regression model by \cite{fritzi2010}. In the context of time series we refer to the articles by \cite{caner2001hansen}, \cite{seo2006}, \cite{seo2008}, and \cite{hansen2002seo}. \cite{lin} considers the adaptive Lasso in a high dimensional quantile threshold model. In panel data, semi-parametrics, and least absolute deviation models, \cite{hansen1999}, \cite{linton2007seo}, \cite{caner2002}, respectively, made contributions. For applications to stock markets and exchange rates we refer to \cite{akdeniz2003} and \cite{basci2006}. These authors argue that threshold model can contribute to reducing forecast errors.  

To be precise, we shall study the model
\begin{equation}
 Y_i = X_i'\beta_0+X_i'\delta_01_{\cbr[0]{Q_i<\tau_0}}+U_i,\qquad i=1,...,n\label{0}
\end{equation}
where $\beta_0,\delta_0\in\mathbb{R}^m$ and $\tau_0$ determines the location of the threshold/break. $Q_i$ determines which regime we are in and could be the debt level in a growth regression or education in a wage regression. If $\delta_0=0$, there is no break and $\tau_0$ is not identified. In that case the model is linear. In a very insightful recent paper \cite{lee2012lasso} proved finite sample oracle inequalities for the prediction and estimation error of the (scaled) Lasso applied to (\ref{0}) in the case of fixed regressors and gaussian error terms. In their simulation section, they also extend their results to random regressors with Gaussian errors. Furthermore, they nicely showed that $\tau_0$ exhibits the well known super efficiency phenomenon from low dimensional break point models even in the high-dimensional case. These authors also show that the scaled Lasso does not select too many irrelevant variables in the spirit of \cite{bickel2009simultaneous}. However, their results are by no means trivial extensions of oracle inequalities for linear models as they show that the classical restricted eigenvalue condition must hold uniformly over the parameter space in threshold models. In addition, the probabilistic analysis is also much more refined than in the linear case. 

The aim of this paper is to show that it is possible to consistently decide whether a break is present or not even in the high-dimensional change point model with random regressors. In other words, we show that it is possible to decide whether $\delta_0=0$ or if it possesses non-zero entries. To do so efficiently, we first establish an upper bound on the sup-norm convergence rate of the estimator $\hat{\delta}$ of $\delta_0$ which is valid in even highly correlated designs. This is not an easy task as almost all previous work has focussed on establishing upper bounds on the $\ell_1$ or $\ell_2$ estimation error in the plain linear model. Exceptions are \cite{lounici2008sup} and \cite{van2014highnotes} who provide sup-norm bounds in the high-dimensional linear model. To the best of our knowledge, we are the first to establish sup-norm bounds on the estimation error in a high-dimensional non-linear model.  Our sup-norm bound is much smaller than the corresponding $\ell_1$ and $\ell_2$ bounds on the estimation error as it does \textit{not} depend on the unknown number of non-zero coefficients $s$. Thus, our approach to break detection, which is based on thresholding, allows for a much finer distinction between zero and non-zero entries of $\delta_0$. The result is that we can detect breaks which would be too small to detect if one thresholded based on classical $\ell_1$ or $\ell_2$ estimation error. In that sense, the sharp sup-norm bound is tailored to break detection in our context and we strengthen the result of selecting not too many irrelevant variables in the threshold model to selecting exactly the right ones with probability tending to one.   

The debate regarding the impact of debt on GDP growth was recently reignited by the European public debt crisis as well the claim by \cite{RR} that public debt has a substantial negative effect on future GDP growth when the ratio of debt to GDP is over 90\%. Following \cite{RR}, several authors have econometrically investigated the presence of such a threshold. Of particular interest for us is the work of \cite{BIS} who estimated threshold growth regressions using several measures of public and private debt as well as a set of standard controls. Using our thresholded Lasso estimator with the data of \cite{BIS} we find robust evidence of a threshold in the effect of debt on future GDP growth. However, the effect of debt being above the threshold appears to be complex.      

In Section \ref{scaledlasso}, we recall the scaled Lasso estimator for threshold models of \cite{lee2012lasso}. Section \ref{uniform} establishes $\ell_{\infty}$ norm bounds for the estimation error of the scaled Lasso. This sup-norm bound is the basis for our new thresholded scaled Lasso estimator which is introduced in Section \ref{thold}. Section \ref{sims} provides simulations supporting the selection consistency of our estimator. Section \ref{application} reports the results of our growth regressions. All proofs are deferred to the appendix.

\subsection{Notation}
For any vector $x\in\mathbb{R}^k$ (for some $k\geq 1$), let $\enVert[0]{x}_{\ell_1}, \enVert[0]{x}_{\ell_2}$ and $\enVert[0]{x}_{\ell_\infty}$ denote the $\ell_1, \ell_2$ and $\ell_\infty$ norms, respectively. Similarly, for any $m\times n$ matrix $A$, $\enVert[0]{A}_{\ell_1}, \enVert[0]{A}_{\ell_2}$ and $\enVert[0]{A}_{\ell_\infty}$ denote the induced (operator) norms corresponding to the above three norms. They can be calculated as $\enVert[0]{A}_{\ell_1}=\max_{1\leq j\leq n}\sum_{i=1}^m|A_{i,j}|$, $\enVert[0]{A}_{\ell_2}=\sqrt{\phi_{\max}(A'A)}$ where $\phi_{\max}(\cdot)$ is the maximal eigenvalue, and $\enVert[0]{A}_{\ell_\infty}=\max_{1\leq i\leq m}\sum_{j=1}^n|A_{i,j}|$, respectively. We will also need $\enVert[0]{A}_\infty=\max_{i,j}|A_{i,j}|$ where the maximum extends over all entries of $A$. For real numbers $a,b$ $a\vee b$ and $a\wedge b$ denote their maximum and minimum, respectively. Furthermore, the empirical norm of $y\in \mathbb{R}^n$ is given by $\enVert[0]{y}_n=\sqrt{\frac{1}{n}\sum_{i=1}^ny_i^2}$. 

We shall say that a real random variable $Z$ is subgaussian if there exists positive constants $A$ and $B$ such that $P\del[0]{|Z|>\tau}\leq Ae^{-Bt^2}$ for all $\tau>0$. $Z$ is said to be subexponential if there exists positive constants $C$ and $D$ such that $P\del[0]{|Z|>\tau}\leq Ce^{-Dt}$ for all $\tau>0$. For $x\in\mathbb{R}^k$, we will let $x^{(j)}$ denote its $j$th entry. Let "wpa1" denote with probability approaching one.

\section{Scaled Lasso for Threshold Regression}\label{scaledlasso}
Defining the $2m\times 1$ vectors $X_i(\tau)=\del[1]{X_i',X_i'1_{\cbr[0]{Q_i<\tau}}}'$ and $\alpha_0=(\beta_0',\delta_0')'$ one can rewrite (\ref{0}) as
\begin{equation}
 Y_i = X_i (\tau_0)' \alpha_0 + U_i,\label{1}\qquad i=1,...,n
\end{equation} 
where $\tau_0$ is supposed to be an element of a parameter space $T=\sbr[0]{t_0,t_1}\subset\mathbb{R}$ and $\alpha_0$ is supposed to belong to a parameter space $\mathcal{A}\subset\mathbb{R}^{2m}$. This is exactly the model that \cite{lee2012lasso} studied in the case where $m$ can be much larger than $n$. We shall be more specific about the probabilistic assumptions in Section \ref{ass}.  
Let $J(\alpha_0)=\cbr[0]{j=1,...,2m:\alpha_0\neq 0}$ be the indices of the non-zero coefficients with cardinality $|J(\alpha_0)|$. Denoting by $X(\tau)$ the $(n\times 2m)$ matrix whose rows are $X_i(\tau)'$, setting $Y=(Y_1,...,Y_n)'$, and $U=(U_1,...,U_n)$, (\ref{1}) can be written more compactly as
\begin{align*}
Y=X(\tau_0)\alpha+U
\end{align*}
Next, let $X^{(j)}(\tau)$ denote the $j$th column of $X(\tau)$ and define the $2m \times 2m$ diagonal matrix
\begin{align*}
D(\tau) = diag \{ \|X^{(j)}(\tau) \|_n, j=1,...,2m\}
\end{align*}
Now set
\begin{align*}
S_n (\alpha, \tau) = n^{-1} \sum_{i=1}^n \del[1]{Y_i - X_i' \beta - X_i' \delta 1_{\{ Q_i < \tau \}}}^2
= 
\| Y- X (\tau) \alpha \|_n^2,
\end{align*}
where $\alpha=(\beta',\delta')'\in\mathcal{A}$ and define the scaled $\ell_1$ penalty
\begin{align*}
\lambda \enVert[1]{D(\tau) \alpha}_{\ell_1} = \lambda \sum_{j=1}^{2m} \|X^{(j)} (\tau) \|_n |\alpha_j|,
\end{align*}
where $\lambda$ is a tuning parameter about which we shall be explicit later. With this notation in place we define for each $\tau\in T$
\begin{align}
\hat{\alpha} (\tau) = \argmin_{\alpha \in A} \cbr[1]{ S_n (\alpha, \tau) + 2 \lambda \enVert[1]{D(\tau) \alpha}_{\ell_1}}\label{s1}
\end{align}
and
\begin{align*}
\hat{\tau} = \argmin_{\tau \in T} \cbr[1]{S_n (\hat{\alpha} (\tau), \tau) + \lambda \enVert[1]{D(\tau) \hat{\alpha} (\tau)}_{\ell_1}}.
\end{align*}
To be precise, $\hat{\tau}$ is an interval and in accordance with \cite{lee2012lasso} we define the maximum of the interval as the estimator $\hat{\tau}$. For every $n$, it suffices in practice to search over $Q_1,...,Q_n$ as candidates for $\hat{\tau}$ as these are the points where $1_{\cbr[0]{Q_i<\tau}},\ i=1,...,n$ can change. Therefore, the estimator of $(\alpha_0, \tau_0)$ is defined as $(\hat{\alpha},\hat{\tau})=(\hat{\alpha}(\hat{\tau}), \hat{\tau})$.
 
Assuming fixed regressors and and gaussian error terms \cite{lee2012lasso}     established oracle inequalities for the prediction and $\ell_1$ estimation error of the Lasso estimator $\hat{\alpha}$. When a break is present they also established upper bounds on the estimation error of $\hat{\tau}$. We contribute by establishing oracle inequalities in the sup-norm for this non-linear model and show that we can consistently detect breaks that are as small as $\sqrt{\frac{\log(m)}{n}}$.

\section{Uniform Convergence Rate of the Scaled Lasso Estimator}\label{uniform}
In this section we establish upper bounds on the sup norm estimation error $\enVert[0]{\hat{\alpha}-\alpha_0}_{\ell_\infty}$. As argued previously, and as will be made rigorous in Section \ref{thold}, an upper bound $\enVert[0]{\hat{\delta}-\delta_0}_{\ell_\infty}$ is what is really needed for break detection purposes. However, we shall actually establish a slightly stronger result here which also makes it possible to efficiently select variables from the first $m$ columns of $X(\tau_0)$. This sup-norm bound is established separately for the case where no break is present and for the case where a break is present.
Let $X$ and $Z(\tau)$ denote the first and last $m$ columns of $X(\tau)$ for $\tau\in T$, respectively, and define
\begin{align*}
r_n=\min_{1\leq j\leq m}\frac{\enVert[1]{Z^{(j)}(t_0)}_n^2}{\enVert[1]{X^{(j)}}_{n}^2}.
\end{align*}
Note that under Assumption 1 below it follows by Lemma \ref{scaling} in the appendix that $r_n$ is bounded away from zero with probability tending to one. $r_n$ is trivially never greater than one. Now define
\begin{align}
\lambda=A\del[3]{\frac{\log(3m)}{nr_n}}^{1/2}\label{lambda}
\end{align}
as the tuning parameter for a constant $A\geq 0$. Assuming an i.i.d. sample we let $\Sigma(\tau)=E\del[1]{X_1(\tau)X_1(\tau)'}$ denote the population covariance matrix of the covariates. In Lemma \ref{ThetaBound1} below we give sufficient conditions for its inverse $\Theta(\tau)$ to exist as long as $\Sigma=E(X_1X_1')$ is invertible which is a standard assumption in regression models. Thus, the practical consequence is that the presence of indicator functions in the definition of $X_1(\tau)$ does not make it singular. Now we introduce the assumptions that our theorems rely on.

\subsection{Assumptions}\label{ass}
In this section we recall the assumptions used by \cite{lee2012lasso} in their Theorems 2 and 3 which are used as ingredients in the proofs of our Theorems \ref{thm1} and \ref{thm2}. To be precise, we use the oracle inequalities for the $\ell_1$ estimation errors of $\hat{\alpha}$ and $\hat{\tau}$ provided by \cite{lee2012lasso}. 
We alter their assumptions slightly, as we are working in a random design as opposed to their fixed regressor design. However, \cite{lee2012lasso} have already argued how some of their assumptions could be valid in a random design and as a consequence we do note need to address these in detail.

\begin{assum}
Let $\cbr[0]{X_i, U_i, Q_i}_{i=1}^n$ be an i.i.d. sample and let $(X_1, U_1)$ be independent of $Q_1$. Furthermore, let $Q_1$ be uniformly distributed on $[0,1]$ and assume that all entries of $X_1$ and $U_1$ are subgaussian\footnote{The notation suppresses that we are really dealing with a triangular array. Thus, more precisely, we assume uniform subgaussianity across the rows of this triangular array.} with $\min_{1\leq j\leq m}E\del[1]{{X_1^{(j)}}^2}$ bounded away from zero.
(i) For the parameter space ${\cal A}$ for $\alpha_0$, any $\alpha \equiv (\alpha_1, \cdots, \alpha_{2m}) \in {\cal A} \subset \mathbb{R}^{2m}$, including $\alpha_0$, satisfies $\max_{1 \le j \le 2m }
|\alpha_j | \le C_1$, for some constant $C_1 >0$. In addition, $\tau_0 \in T=[t_0, t_1]$ with $0<t_0<t_1<1$. $(ii) \log(m)/n\to 0$.
\end{assum}

Assumption 1 is the one which has been altered the most compared to \cite{lee2012lasso} as the boundedness of certain norms of the covariates does no longer have to be assumed as this now follows directly from independence and subgaussianity of these. See Lemma \ref{scaling} in the appendix for details. Furthermore, the absence of ties among the $Q_i,\ i=1,..., n$ (as required in \cite{lee2012lasso}) follows in an almost sure sense from these being uniformly (and thus continuously) distributed. 

The assumption of the sample being i.i.d. can most likely be relaxed by exchanging the probabilisitic inequalities used in the appendix for ones allowing for weak dependences and/or heterogeneity. For convenience, we have also assumed that $X_1$ and $Q_1$ are independent. However, as the main contribution of this paper is to provide sup norm bounds for high-dimensional non-linear models as the first in the literature (to the best of our knowledge) we have chosen to keep the probabilistic framework simple in order not to suffocate the cardinal ideas in technicalities.


\begin{assum}
(Uniform Restricted Eigenvalue Condition). For some integer $s$ such that $1 \le s \le 2m$, a positive number $c_0$ and some set ${\cal S} \subset \mathbb{R}$, the following 
condition holds wpa1
\begin{align}
\kappa (s, c_0, {\cal S}) = \min_{\tau \in {\cal S}} \quad \min_{J_0 \subset \{1,..., 2m \}, |J_0 | \le s} \quad  \min_{\gamma \neq 0, |\gamma_{J_0}^c |_1 \le c_0 |\gamma_{J_0}|_1} \frac{ |X (\tau) \gamma |_2}{n^{1/2} |\gamma_{J_0}|_2}
>0.\label{recond}
\end{align}
\end{assum}
In the random design considered in this paper we require assumption 2 of    \cite{lee2012lasso} above to be valid with probability tending to one. However, this is an unnecessarily high-level assumption as it can often be verified by assuming that $\Sigma(\tau)$ satisfies the uniform restricted eigenvalue condition (which it does in particular when it has full rank --  as is in turns true under Assumption 1 if $\Sigma$ has full rank as argued on page A4 in \cite{lee2012lasso}) and by showing that $\frac{1}{n}X'(\tau)X(\tau)$ is uniformly close to $\Sigma(\tau)$. Mimicking the arguments on pages A3-A6 in \cite{lee2012lasso} it can be shown that (\ref{recond}) above holds with probability tending to one under our Assumption 1 as long as $\Sigma$ has full rank -- a rather innocent assumption. Thus, Assumption 2 is almost automatic under Assumption 1 and we shall use this in the statements of Theorems \ref{thm1} and \ref{thm2} below.

For the next assumption, define $f_{\alpha, \tau}(x,q) = x' \beta + x' \delta 1_{ \{ q < \tau\}}$, and $f_0(x,q) = x'\beta_0 + x' \delta_0 1_{ \{ q < \tau_0 \} }$ and let $m(\alpha)$ denote the number of non-zero elements of $\alpha$.

\begin{assum}
(Identifiability under Sparsity and Discontinuity of Regression). For a given $s \geq |J(\alpha_0)|$, and for any $\eta$ and $\tau$ such that $|\tau - \tau_0 | > \eta \ge \min_i |Q_i - \tau_0|$, and 
$ \alpha \in \{ \alpha: m(\alpha) \le s\}$ there exists a constant  $c>0$ such that, wpa1
\[ \| f_{\alpha, \tau} - f_0 \|_n^2 > c \eta ,\] 
\end{assum}

For this assumption \cite{lee2012lasso} (pages A7-A8) also provide sufficient conditions encompassing the assumptions made in Assumption 1 above.

\begin{assum}
(Smoothness of Design). For any $\eta > 0$, there exists a constant $C<\infty$ such that wpa1
\[ \sup_{1\leq j,k\leq m} \sup_{|\tau - \tau_0| < \eta } \frac{1}{n} \sum_{i=1}^n \envert[1]{X_i^{(j)}X_i^{(k)}} |1_{\{Q_i < \tau_0\}} - 1_{ \{ Q_i < \tau \}} | \le C \eta.\]
\end{assum}
\cite{lee2012lasso} argue that this is the case when the $Q_i$ are continuously distributed and $E\del[1]{\envert[1]{X_i^{(j)}X_i^{(k)}}|Q_i=\tau}$ is continuous and bounded in a neighborhood of $\tau_0$ for all $1\leq j,k\leq m$.
Note however, that the outer supremum in Assumption 4 above is taken over all $1\leq j,k\leq m$ as opposed to only $1\leq j\leq m$  in \cite{lee2012lasso} as $\envert[1]{X_i^{(j)}X_i^{(k)}} $ has replaced ${X_i^{(j)}}^2$. This slight strengthening of the assumption is needed to establish an $\ell_\infty$ bound on the estimation error of $\hat{\alpha}$ in the case where a structural break is present (Theorem \ref{thm2} below).

\begin{assum}
(Well defined second moments). For any $\eta$ such that $1/n \le \eta \le \eta_0$, $h_n^2(\eta)$ is bounded where wpa1
\[ h_n^2(\eta) = \frac{1}{2 n \eta} \sum_{i=\max\{1, [n(\tau_0 - \eta)]\}}^{\min\{ [n(\tau_0+\eta)],n\}} (X_i' \delta_0)^2,\]
where $[.]$ denotes the integer part of a real number. 
\end{assum}
Finally, we also need to impose the same technical regularity condition as \cite{lee2012lasso} which they denote Assumption 6 and present on page A23 of their paper. This assumption is satisfied asymptotically in our context when $s\enVert[0]{\delta_0}_{\ell_1}\sqrt{\frac{\log(m)}{n}}\to 0$. Since $\max_{1\leq j\leq m}\delta_{0,j}\leq C_1$ by Assumption 1 above this is in turns true when $s\envert[1]{J(\delta_0)}\log(m)^{1/2}/\sqrt{n}\to 0$. The latter assumption will be assumed in Theorem \ref{thm2} below (as we also need it for another purpose) and thus Assumption 6 in \cite{lee2012lasso} is automatic in our case.

\subsection{sup-norm rate of convergence of $\hat{\alpha}$}
We next turn to providing upper bounds on the $\ell_\infty$ estimation error of $\hat{\alpha}$. We distinguish between the case in which no break is present and the case in which a break is present.
\begin{thm}\label{thm1}
Suppose that $\delta_0=0$ and let Assumptions 1 be satisfied. Furthermore, let $|J(\alpha)|\leq s$, assume that $\Sigma$ has full rank and that $\Theta(\tau)=\Sigma^{-1}(\tau)$ satisfies $\sup_{\tau\in T}\enVert[1]{\Theta(\tau)}_{\ell_\infty}<\infty$. Then, choosing $\lambda$ as in (\ref{lambda}) and assuming $s\sqrt{\frac{\log(mn)}{n}}\to 0$, one has
\begin{align*}
\enVert[1]{\hat{\alpha}-\alpha_0}_{\ell_\infty}=O_p\del[3]{\sqrt{\frac{\log(m)}{n}}}=O_p(\lambda).
\end{align*}
Thus, a fortiori, we also have $\enVert[1]{\hat{\delta}-\delta_0}_{\ell_\infty}=O_p\del[3]{\sqrt{\frac{\log(m)}{n}}}=O_p(\lambda)$.
\end{thm}
Theorem \ref{thm1} provides the stochastic order of the $\ell_\infty$ estimation error of $\hat{\alpha}$ for the case where no break is present. From Theorem 1 in \cite{lee2012lasso} (ignoring that their results are for non-random regressors) one can conclude that $\enVert[0]{\hat{\alpha}-\alpha_0}_{\ell_1}=O_p\del[1]{s\sqrt{\log(m)/n}}$. From this, one can of course also conclude that $\enVert[0]{\hat{\alpha}-\alpha_0}_{\ell_\infty}\leq \enVert[0]{\hat{\alpha}-\alpha_0}_{\ell_1}=O_p\del[1]{s\sqrt{\log(m)/n}}$. However, our Theorem \ref{thm1} shows that this rate is much too large as $s$ may be almost as large as $O(\sqrt{n})$ without obstructing $\ell_1$ norm consistency. Our much smaller bound will allow for more precise thresholding in Section \ref{thold}.

We stress again that almost all research in high-dimensional models so far has focussed exclusively on providing upper bounds on the $\ell_1$ and $\ell_2$. $\ell_\infty$ bounds on the estimation error have been established for the Lasso in the plain linear regression model by \cite{lounici2008sup} and \cite{van2014highnotes}. However, to the best of our knowledge we are the first to establish sup-norm bounds for high-dimensional non-linear models, and certainly in the threshold model. As we shall see below, a sup-norm bound will yield much more precise variable selection results for the thresholded scaled Lasso than thresholding based on $\ell_1$ or $\ell_2$ bounds since the latter two are larger due to the presence of the unknown sparsity $s$. Next, consider the case where $\delta_0\neq 0$, i.e. a break is present.

\begin{thm}\label{thm2}
Suppose that $\delta_0\neq 0$ and let Assumptions 1 and 3-5 be satisfied. Furthermore, let $|J(\alpha)|\leq s$, assume that $\Sigma$ has full rank and that $\enVert[1]{\Theta(\tau_0)}_{\ell_\infty}<\infty$. Then, choosing $\lambda$ as in (\ref{lambda}) and assuming $s\envert[1]{J(\delta_0)}\sqrt{\frac{\log(m)}{n}}\to 0$, one has
\begin{align*}
\enVert[1]{\hat{\alpha}-\alpha_0}_{\ell_\infty}=O_p\del[2]{\sqrt{\frac{\log(m)}{n}}}.
\end{align*}
Thus, a fortiori, we also have $\enVert[1]{\hat{\delta}-\delta_0}_{\ell_\infty}=O_p\del[3]{\sqrt{\frac{\log(m)}{n}}}=O_p(\lambda)$.
\end{thm}
The results of Theorem \ref{thm2} are similar to those in Theorem \ref{thm1} but the assumptions differ. First, $\enVert[1]{\Theta(\tau)}_{\ell_\infty}$ only has to be bounded at $\tau_0$ instead of uniformly over $T=[t_0,t_1]$ for $0<t_0<t_1<1$.  Lemma \ref{ThetaBound1} below shows that $\sup_{\tau\in T}\enVert[1]{\Theta(\tau)}_{\ell_\infty}<\infty$ and $\enVert[1]{\Theta(\tau_0)}_{\ell_\infty}<\infty$ in the equicorrelation design but of course with the former being no smaller than the latter. More importantly, requiring $s\envert[1]{J(\delta_0)}\log(m)^{1/2}/\sqrt{n}\to 0$ is in general more restrictive than requiring $s\sqrt{\frac{\log(mn)}{n}}\to 0$ as in Theorem \ref{thm1}. However, if the number of coefficient which break is bounded, i.e. $\envert[1]{J(\delta_0)}\leq B$ for an absolute constant $B$, then the rate requirement of Theorem \ref{thm2} is actually slightly weaker than the one in Theorem \ref{thm1}.

The following Lemma shows that even when the covariates are highly correlated, $\Sigma^{-1}$ exists and the assumptions $\sup_{\tau\in T}\enVert[1]{\Theta(\tau)}_{\ell_\infty}<\infty$ and $\enVert[1]{\Theta(\tau_0)}_{\ell_\infty}<\infty$ from Theorems \ref{thm1} and \ref{thm2}, respectively, are satisfied. First, recall the definition of an equicorrelation design.
\begin{definition}
We say that $\Sigma$ is an equicorrelation matrix if
\begin{align*}
\Sigma =
 \begin{pmatrix}
  1 & \rho & \cdots & \rho \\
  \rho & 1 & \cdots & \rho \\
  \vdots  & \vdots  & \ddots & \vdots  \\
  \rho & \rho & \cdots & 1
 \end{pmatrix}
\end{align*}
for some $-1<\rho<1$.
\end{definition}

\begin{lemma}\label{ThetaBound1}
Let $\cbr[1]{X_i,U_i}_{i=1}^n$ be an iid sample and assume that $U_1$ is uniformly distributed on $[0,1]$. Let $\Sigma=E(X_1X_1')$ be an $m\times m$ equicorrelation matrix with $0\leq\rho<1$. Then $\Sigma^{-1}$ exists and for all $\tau\in (0,1)$ one has $\enVert[1]{\Theta(\tau)}_{\ell_\infty}\leq \frac{2}{(1-\tau)(1-\rho)}\del[1]{2\vee \frac{\tau+1}{\tau}}$. If, furthermore, $T=[t_0,t_1]$ for some $0<t_0<t_1<1$, then $\sup_{\tau\in T}\enVert[1]{\Theta(\tau)}_{\ell_\infty}$ is bounded by a constant only depending on $\rho$. 
\end{lemma}
 Lemma \ref{ThetaBound1} states that $\enVert[1]{\Theta(\tau)}_{\ell_\infty}$ is bounded for all $\tau\in(0,1)$ even when the correlation is arbitrarily close to, but different from, one. $\tau$ can not be zero or one since in that case $\Sigma(\tau)$ would be singular. From a modeling point of view this excludes breaks at the very endpoints of the sample which is a standard assumption in the literature.
 
\section{Thresholded Scaled Lasso}\label{thold}
In this section we utilize the $\ell_\infty$ bound established in Theorems \ref{thm1} and \ref{thm2} above to provide sharp thresholding results for the Scaled Lasso estimator. Recall that these theorems established that $\enVert[1]{\hat{\alpha}-\alpha_0}_{\ell_\infty}\leq C\lambda$ with arbitrarily large probability, irrespective of whether a break is present or not, by choosing $C$ sufficiently large. Before showing that the breaks can be revealed consistently we shall provide a slightly more general result stating that the truly zero coefficients can be distinguished from the non-zero ones. First, define the Thresholded Scaled Lasso estimator as 
\begin{align}
\tilde{\alpha}_j
=
\begin{cases} 
\hat{\alpha}_j \qquad &\text{if} \qquad |\hat{\alpha}_j|\geq H\\
0\qquad &\text{if} \qquad |\hat{\alpha}_j|< H\end{cases}\label{thresh}
\end{align}
where $H$ is the threshold determining whether a coefficient should be classified as zero or non-zero. In particular, we shall see that choosing $H=2C\lambda$ results in consistent model selection. Here we stress once more that our threshold is much sharper than what would have been obtainable if we had directly used that $\enVert[1]{\hat{\alpha}-\alpha_0}_{\ell_1}\leq Cs\lambda$ with probability tending to one from \cite{lee2012lasso}. Thus, it is important to have an $\ell_\infty$ bound on the estimation error as this allows for a much finer distinction between the zero and the non-zero coefficients than would been possible from the usual $\ell_1$ or $\ell_2$ bounds. To be precise, let $\alpha_{0j}$ be a nonzero coefficient such that $|\alpha_{0j}|/\lambda\to \infty$ but $|\alpha_{0j}|/(s\lambda)\to 0$. Not that there may be a considerable wedge between $|\alpha_{0j}|/\lambda$ and $|\alpha_{0j}|/(s\lambda)$ as $s$ can be almost as large as $\sqrt{n}$ such that this is a setting of practical relevance. Such an $\alpha_{0,j}$ will correctly be classified as non-zero when thresholding at the level $\lambda$ (resulting from an $\ell_\infty$ bound) while it would wrongly be classified as zero when thresholding at the level $s\lambda$ (resulting from a plain $\ell_1$ bound). This example underscores the importance of establishing $\ell_\infty$ bounds as in Theorems \ref{thm1} and \ref{thm2} prior to thresholding. Next, recall that $J(\alpha_0) = \{ j=1,...,2m: \alpha_{0j} \neq 0 \}$ and define $J(\tilde{\alpha})=\cbr[0]{j=1,....,2m: \tilde{\alpha}_{j}\neq 0}$. The following  theorems establish the properties of the thresholded scaled Lasso and rely crucially on the $\ell_\infty$ bounds on the estimation error established in Theorems \ref{thm1} and \ref{thm2} above.  
\begin{thm}\label{thmthresh}
Let the assumptions of Theorems \ref{thm1} and \ref{thm2} be satisfied and assume that $\min_{ j \in J(\alpha_0)} |\alpha_{0j}| > 3 C \lambda$. Then, for all $\epsilon>0$ there exists a $C$ such that for $H=2C\lambda=2C\sqrt{\frac{\log(m)}{n}}$ one has $P\del[1]{J(\tilde{\alpha})=J(\alpha_0)}\geq 1-\epsilon$.
\end{thm}
Theorem \ref{thmthresh} states that consistent model selection is possible with the thresholded Lasso in the nonlinear break point regression model as long as the non-zero coefficients are at least of the order $\sqrt{\frac{\log(m)}{n}}$. This is considerably sharper than thresholding based on $\ell_1$ estimation errors where consistent variable selection would require the non-zero coefficients to be at least of order $s\sqrt{\frac{\log(m)}{n}}$. The idea in the proof of Theorem \ref{thmthresh} is similar to the one for the linear case in \cite{lounici2008sup}.

Note that if one is only interested in finding out whether there is a break or not, i.e. whether $\delta_0$ is non-zero or not, one can simply threshold $\tilde\delta$ only according to the rule in (\ref{thresh}). Defining $J({\delta_0}) = \{ j=1,...,m: \delta_{0j} \neq 0 \}$ and $J(\tilde{\delta})=\cbr[0]{j=1,....,m: \tilde{\delta}_{j}\neq 0}$ we have the following result on consistent break detection.
\begin{thm}\label{thmbreak}
Let the assumptions of Theorems \ref{thm1} and \ref{thm2} be satisfied and assume that $\min_{ j \in J(\delta_0)} |\delta_{0j}| > 3 C \lambda$. Then, for all $\epsilon>0$ there exists a $C$ such that for $H=2C\lambda=2C\sqrt{\frac{\log(m)}{n}}$ one has $P\del[1]{J(\tilde{\delta})=J(\delta_0)}\geq 1-\epsilon$.
\end{thm}
Break selection consistency is weaker than model selection consistency as it only requires classifying $\delta_0$ correctly. However, it is still relevant as it answers the question whether a break is present or not. We discuss how to choose the threshold parameter $C$ in practice in Section \ref{sims}.

\section{Simulations}\label{sims}

In this section we report the results of a series of simulation experiments evaluating the finite sample properties of the thresholded scaled Lasso. We shall consider performance along the dimensions: increasing number of irrelevant variables, estimation in the absence of a threshold, increasing number of observations, scale of the parameters, and increasing number of non-zero variables.

The regressors are generated as $X_i\sim\mathcal{N}(0,I)$, the threshold variable $Q_i\sim\mathcal{U}[0,1]$, and the innovations $U_i\sim\mathcal{N}(0,\sigma^2)$ where we set the residual variance $\sigma^2=0.25,\ i=1,...,n$.
When the threshold parameter $\tau_0$ is not explicitly stated it is set to $\tau_0=0.5$; we search for $\tau_0$ over a grid from $0.15$ to $0.85$ by steps of $0.05$. This grid is coarser than the grid used in \cite{lee2012lasso} which, in our experience, has a mild detrimental effect on the precision with which $\tau_0$ is estimated but not on other measures of the quality of the estimator while substantially reducing computation time, thus allowing us to carry out more replications. We select the thresholding parameter $C$ by BIC using a grid from $0.1$ to $5$, so that parameters smaller (in absolute value) than $\widehat C \widehat \lambda$ are set to zero by the thresholded scaled Lasso.

Every model is estimated with an intercept so that we estimate $2m+1$ parameters, plus the threshold parameter $\tau_0$.
All the results reported below are based on 1000 replications.
The simulation are carried with \texttt{R} \citep{R} using the \texttt{glmnet} package of \cite{glmnet}. The results (and those of the empirical application in section \ref{application}) can be replicated using \texttt{knitr} \citep{knitr} and the supplementary material\footnote{Available at \url{https://github.com/lcallot/ttlas}}.

We report the following statistics, averaged across iterations. 
\begin{itemize}
	\item MSE: mean square prediction error.
	\item $| J(\hat \alpha) \cap J(\alpha_0)^c|$: number zero parameters incorrectly retained in the model.
	\item $|J(\alpha_0)\cap J(\hat \alpha)^c|$: number of non-zero parameters excluded.
	\item Perfect Sel.: the share (in \%) of iterations for which we have perfect model selection.
	\item $\enVert{\hat\alpha - \alpha_0}_1$: $\ell_1$ estimation error for the parameters.
	\item $\enVert{\hat\alpha - \alpha_0}_{\infty}$: $\ell_\infty$ estimation error for the parameters.
	\item $|\hat\tau - \tau_0|$: absolute threshold parameter estimation error. 
	\item C: selected (BIC) thresholding parameter.
	\item $\hat{\lambda}$: selected (BIC) penalty parameter.
\end{itemize}

\begin{table}[!htpb]
\centering
\begin{tabular}{ l l H r r r r H r r r r r }
\toprule
 & $\tau_0$ &  & \rotpar{ MSE } & \rotpar{ {\small $ |J(\hat \alpha) \cap J(\alpha_0)^c|$} } & \rotpar{ {\small $|J(\alpha_0) \cap J(\hat\alpha)^c|$} } & \rotpar{ Perfect Sel } & \rotpar{ \# non Zeros } & \rotpar{ $\enVert{\hat\alpha - \alpha_0}_1$ } & \rotpar{ $\enVert{\hat\alpha - \alpha_0}_{\infty}$ } & \rotpar{ $|\hat\tau - \tau_0|$ } & C & $\hat\lambda$\\
\midrule
 \cellcolor{white}  & \cellcolor{white}  &  L & 1.20 & 4.50 & 0.06 & 3 & 14.44 & 3.39 & 0.85 & 0.27 &  -  & 0.06\\
 \rowcolor{black!10} \cellcolor{white}  & \cellcolor{white} \multirow{-2}{*}{0.3} &  T & 1.22 & 0.04 & 0.11 & 86 & 9.93 & 3.29 & 0.85 &  -  & 1.46 &  - \\
 \cellcolor{white}  & \cellcolor{white}  &  L & 1.42 & 5.23 & 0.07 & 1 & 15.16 & 3.68 & 0.94 & 0.22 &  -  & 0.05\\
 \rowcolor{black!10} \cellcolor{white}  & \cellcolor{white} \multirow{-2}{*}{0.4} &  T & 1.45 & 0.04 & 0.16 & 81 & 9.88 & 3.55 & 0.94 &  -  & 1.49 &  - \\
 \cellcolor{white}  & \cellcolor{white}  &  L & 1.55 & 5.72 & 0.07 & 0 & 15.66 & 3.99 & 1.00 & 0.18 &  -  & 0.05\\
 \rowcolor{black!10} \cellcolor{white} \multirow{-6}{*}{$m=50$} & \cellcolor{white} \multirow{-2}{*}{0.5} &  T & 1.58 & 0.06 & 0.14 & 82 & 9.91 & 3.85 & 1.01 &  -  & 1.45 &  - \\
\cmidrule(l){1 - 13}
 \cellcolor{white}  & \cellcolor{white}  &  L & 1.34 & 5.59 & 0.05 & 1 & 15.53 & 3.99 & 0.95 & 0.25 &  -  & 0.07\\
 \rowcolor{black!10} \cellcolor{white}  & \cellcolor{white} \multirow{-2}{*}{0.3} &  T & 1.38 & 0.04 & 0.13 & 85 & 9.91 & 3.86 & 0.95 &  -  & 1.27 &  - \\
 \cellcolor{white}  & \cellcolor{white}  &  L & 1.56 & 6.26 & 0.08 & 0 & 16.18 & 4.29 & 1.03 & 0.22 &  -  & 0.07\\
 \rowcolor{black!10} \cellcolor{white}  & \cellcolor{white} \multirow{-2}{*}{0.4} &  T & 1.60 & 0.05 & 0.16 & 82 & 9.89 & 4.15 & 1.03 &  -  & 1.25 &  - \\
 \cellcolor{white}  & \cellcolor{white}  &  L & 1.77 & 7.27 & 0.12 & 0 & 17.15 & 4.77 & 1.10 & 0.19 &  -  & 0.07\\
 \rowcolor{black!10} \cellcolor{white} \multirow{-6}{*}{$m=100$} & \cellcolor{white} \multirow{-2}{*}{0.5} &  T & 1.83 & 0.07 & 0.21 & 78 & 9.86 & 4.60 & 1.11 &  -  & 1.22 &  - \\
\cmidrule(l){1 - 13}
 \cellcolor{white}  & \cellcolor{white}  &  L & 1.57 & 7.06 & 0.10 & 0 & 16.95 & 4.65 & 1.06 & 0.25 &  -  & 0.09\\
 \rowcolor{black!10} \cellcolor{white}  & \cellcolor{white} \multirow{-2}{*}{0.3} &  T & 1.62 & 0.03 & 0.19 & 82 & 9.84 & 4.49 & 1.06 &  -  & 1.15 &  - \\
 \cellcolor{white}  & \cellcolor{white}  &  L & 1.80 & 8.10 & 0.12 & 0 & 17.97 & 5.04 & 1.14 & 0.22 &  -  & 0.09\\
 \rowcolor{black!10} \cellcolor{white}  & \cellcolor{white} \multirow{-2}{*}{0.4} &  T & 1.87 & 0.03 & 0.22 & 79 & 9.81 & 4.86 & 1.15 &  -  & 1.12 &  - \\
 \cellcolor{white}  & \cellcolor{white}  &  L & 2.22 & 9.20 & 0.26 & 0 & 18.94 & 5.82 & 1.27 & 0.18 &  -  & 0.09\\
 \rowcolor{black!10} \cellcolor{white} \multirow{-6}{*}{$m=200$} & \cellcolor{white} \multirow{-2}{*}{0.5} &  T & 2.30 & 0.06 & 0.40 & 71 & 9.67 & 5.60 & 1.28 &  -  & 1.07 &  - \\
\cmidrule(l){1 - 13}
 \cellcolor{white}  & \cellcolor{white}  &  L & 1.73 & 8.81 & 0.15 & 0 & 18.66 & 5.38 & 1.16 & 0.26 &  -  & 0.10\\
 \rowcolor{black!10} \cellcolor{white}  & \cellcolor{white} \multirow{-2}{*}{0.3} &  T & 1.81 & 0.03 & 0.23 & 81 & 9.81 & 5.18 & 1.17 &  -  & 1.04 &  - \\
 \cellcolor{white}  & \cellcolor{white}  &  L & 2.16 & 9.35 & 0.33 & 0 & 19.02 & 6.17 & 1.30 & 0.22 &  -  & 0.12\\
 \rowcolor{black!10} \cellcolor{white}  & \cellcolor{white} \multirow{-2}{*}{0.4} &  T & 2.26 & 0.04 & 0.47 & 73 & 9.57 & 5.94 & 1.31 &  -  & 0.98 &  - \\
 \cellcolor{white}  & \cellcolor{white}  &  L & 2.84 & 9.81 & 0.66 & 0 & 19.15 & 7.26 & 1.46 & 0.19 &  -  & 0.13\\
 \rowcolor{black!10} \cellcolor{white} \multirow{-6}{*}{$m=400$} & \cellcolor{white} \multirow{-2}{*}{0.5} &  T & 2.96 & 0.03 & 0.91 & 60 & 9.12 & 7.02 & 1.47 &  -  & 0.90 &  - \\
\bottomrule
\end{tabular}\caption{Lasso (white background) and Thresholded Lasso (grey background). Increasing number of zero parameters and 3 locations of $\tau_0$.}
\label{tab:xpzero}
\end{table}

Table \ref{tab:xpzero} contains the results of experiments where we consider 4 different dimensions for the parameter vectors and multiple locations for $\tau_0$. The data is generated as follows:
\begin{itemize}
	\item Sample size: $n=200$, $\beta = [2,2,2,2,2,0,...,0]$, $\delta = [2,-2,2,-2,2,0,...,0]$.
	\item The length $\beta$ and $\delta$ is $m=50,100,200,400$.
\end{itemize}

The most important finding in Table \ref{tab:xpzero} is that across all settings the scaled Lasso almost never detects the true model while its thresholded version does so very often and rather consistent across the settings. As expected, the scaled Lasso does a good job at model screening in the sense that it retains all relevant variables in many instances. However, it often fails to exclude irrelevant variables. This is exactly where the thresholding sets in --  it weeds out the falsely retained variables by the first step scaled Lasso. To illustrate this, consider the setting of $m=400$ and $\tau_0=0.5$. Here the scaled Lasso includes almost ten irrelevant variables on average while its thresholded version includes as few as 0.03 irrelevant variables on average. Note also how the $\ell_\infty$ estimation error is much lower than the $\ell_1$ counterpart confirming our theoretical results from Theorems \ref{thm1} and \ref{thm2}, thus allowing for much sharper thresholding than usual. This important finding is confirmed in all of the other settings below.

Perfect model selection seems to be slightly easier for lower values of the threshold parameter $\tau_0$. On the other hand, $\hat{\tau}$ becomes less precise as $\tau_0$ is lowered. All other measures in general improve slightly when $\tau_0$ is lowered. Increasing the dimension of the model, $m$, worsens most performance measures except for the estimation error of $\hat{\tau}$ which stays constant.  Finally, in larger models more penalization is applied as can be seen from the larger choice of $\lambda$ as $m$ is increased.

Table \ref{tab:xpnj} considers the case where no threshold effect is present, $\delta_0=0$, the exact data generating process is:
\begin{itemize}
	\item Sample size: $n=200$, $\beta = [2,2,2,2,2,0,...,0]$, $\delta = [0,...,0]$.
	\item The length $\beta$ and $\delta$ is $m=50,100,200,400$.
\end{itemize}

\begin{table}[!htpb]
\centering
\begin{tabular}{ l H H r r r r H r r H r r }
\toprule
 &  &  & \rotpar{ MSE } & \rotpar{ {\small $ |J(\hat \alpha) \cap J(\alpha_0)^c|$} } & \rotpar{ {\small $|J(\alpha_0) \cap J(\hat\alpha)^c|$} } & \rotpar{ Perfect Sel } & \rotpar{ \# non Zeros } & \rotpar{ $\enVert{\hat\alpha - \alpha_0}_1$ } & \rotpar{ $\enVert{\hat\alpha - \alpha_0}_{\infty}$ } & \rotpar{ $|\hat\tau - \tau_0|$ } & C & $\hat\lambda$\\
\midrule
 \cellcolor{white}  & \cellcolor{white}  &  L & 0.29 & 1.56 & 0.00 & 23 & 6.56 & 0.60 & 0.16 & 0.26 &  -  & 0.07\\
 \rowcolor{black!10} \cellcolor{white} \multirow{-2}{*}{$m=50$} & \cellcolor{white} \multirow{-2}{*}{$n=200$} &  T & 0.29 & 0.21 & 0.00 & 81 & 5.21 & 0.56 & 0.16 &  -  & 0.73 &  - \\
\cmidrule(l){1 - 13}
 \cellcolor{white}  & \cellcolor{white}  &  L & 0.30 & 1.56 & 0.00 & 23 & 6.57 & 0.65 & 0.17 & 0.27 &  -  & 0.08\\
 \rowcolor{black!10} \cellcolor{white} \multirow{-2}{*}{$m=100$} & \cellcolor{white} \multirow{-2}{*}{$n=200$} &  T & 0.31 & 0.18 & 0.00 & 83 & 5.18 & 0.61 & 0.17 &  -  & 0.61 &  - \\
\cmidrule(l){1 - 13}
 \cellcolor{white}  & \cellcolor{white}  &  L & 0.31 & 1.45 & 0.00 & 27 & 6.45 & 0.70 & 0.18 & 0.27 &  -  & 0.09\\
 \rowcolor{black!10} \cellcolor{white} \multirow{-2}{*}{$m=200$} & \cellcolor{white} \multirow{-2}{*}{$n=200$} &  T & 0.32 & 0.15 & 0.00 & 86 & 5.15 & 0.66 & 0.18 &  -  & 0.53 &  - \\
\cmidrule(l){1 - 13}
 \cellcolor{white}  & \cellcolor{white}  &  L & 0.32 & 1.44 & 0.00 & 27 & 6.43 & 0.74 & 0.19 & 0.27 &  -  & 0.10\\
 \rowcolor{black!10} \cellcolor{white} \multirow{-2}{*}{$m=400$} & \cellcolor{white} \multirow{-2}{*}{$n=200$} &  T & 0.33 & 0.12 & 0.00 & 89 & 5.12 & 0.71 & 0.19 &  -  & 0.46 &  - \\
\bottomrule
\end{tabular}
\caption{Lasso (white background) and Thresholded Lasso (grey background). No threshold effect ($\delta=0$), $n=200$, 4 different length of the parameter vector.}
\label{tab:xpnj}
\end{table}

The main finding of Table \ref{tab:xpnj} is that almost all performance measures improve drastically compared to Table \ref{tab:xpzero}. This is the case in particular for large $m$ as the performance is no longer worsened as $m$ increases. Note, for example, that the MSE and $\ell_1$ estimation error of $\hat{\alpha}$ are almost ten times lower for $m=400$ than they were in Table \ref{tab:xpzero}. Most importantly for us, the perfect models selection percentage is now also stable across $m$.

In order to investigate the asymptotic properties of our procedure, Table \ref{tab:xpsmpl} reveals the effect of increasing the sample size for two values of $\tau_0$. The exact DGP is: 
\begin{itemize}
	\item Sample size: $n=50,100,200,500,1000$.
	\item $\beta = [2,2,2,2,2,0,...,0]$, $\delta = [2,-2,2,-2,2,0,...,0]$.
\end{itemize}

\begin{table}[!htpb]
\centering
\begin{tabular}{ l l H r r r r H r r r r r }
\toprule
 &  &  & \rotpar{ MSE } & \rotpar{ {\small $ |J(\hat \alpha) \cap J(\alpha_0)^c|$} } & \rotpar{ {\small $|J(\alpha_0) \cap J(\hat\alpha)^c|$} } & \rotpar{ Perfect Sel } & \rotpar{ \# non Zeros } & \rotpar{ $\enVert{\hat\alpha - \alpha_0}_1$ } & \rotpar{ $\enVert{\hat\alpha - \alpha_0}_{\infty}$ } & \rotpar{ $|\hat\tau - \tau_0|$ } & C & $\hat\lambda$\\
\midrule
 \cellcolor{white}  & \cellcolor{white}  &  L & 10.04 & 1.83 & 4.92 & 0 & 6.91 & 14.72 & 1.99 & 0.30 &  -  & 0.58\\
 \rowcolor{black!10} \cellcolor{white}  & \cellcolor{white} \multirow{-2}{*}{$n=50$} &  T & 10.64 & 0.29 & 5.51 & 0 & 4.78 & 14.66 & 1.99 &  -  & 0.67 &  - \\
 \cellcolor{white}  & \cellcolor{white}  &  L & 3.34 & 7.22 & 1.09 & 0 & 16.13 & 7.92 & 1.51 & 0.27 &  -  & 0.15\\
 \rowcolor{black!10} \cellcolor{white}  & \cellcolor{white} \multirow{-2}{*}{$n=100$} &  T & 3.53 & 0.12 & 1.38 & 45 & 8.74 & 7.63 & 1.51 &  -  & 1.32 &  - \\
 \cellcolor{white}  & \cellcolor{white}  &  L & 1.46 & 5.56 & 0.08 & 1 & 15.48 & 4.07 & 1.00 & 0.25 &  -  & 0.07\\
 \rowcolor{black!10} \cellcolor{white}  & \cellcolor{white} \multirow{-2}{*}{$n=200$} &  T & 1.50 & 0.04 & 0.16 & 82 & 9.87 & 3.95 & 1.00 &  -  & 1.25 &  - \\
 \cellcolor{white}  & \cellcolor{white}  &  L & 0.76 & 3.31 & 0.01 & 6 & 13.30 & 2.27 & 0.64 & 0.17 &  -  & 0.04\\
 \rowcolor{black!10} \cellcolor{white}  & \cellcolor{white} \multirow{-2}{*}{$n=500$} &  T & 0.76 & 0.01 & 0.02 & 97 & 9.99 & 2.23 & 0.64 &  -  & 0.95 &  - \\
 \cellcolor{white}  & \cellcolor{white}  &  L & 0.50 & 2.62 & 0.00 & 10 & 12.62 & 1.51 & 0.45 & 0.06 &  -  & 0.03\\
 \rowcolor{black!10} \cellcolor{white} \multirow{-10}{*}{$\tau_0= 0.3$} & \cellcolor{white} \multirow{-2}{*}{$n=1000$} &  T & 0.50 & 0.00 & 0.01 & 98 & 10.00 & 1.49 & 0.45 &  -  & 0.81 &  - \\
\cmidrule(l){1 - 13}
 \cellcolor{white}  & \cellcolor{white}  &  L & 8.98 & 1.81 & 4.84 & 0 & 6.98 & 14.56 & 2.00 & 0.21 &  -  & 0.48\\
 \rowcolor{black!10} \cellcolor{white}  & \cellcolor{white} \multirow{-2}{*}{$n=50$} &  T & 9.52 & 0.24 & 5.43 & 0 & 4.81 & 14.48 & 2.00 &  -  & 0.62 &  - \\
 \cellcolor{white}  & \cellcolor{white}  &  L & 4.73 & 5.41 & 2.15 & 0 & 13.26 & 10.05 & 1.75 & 0.20 &  -  & 0.21\\
 \rowcolor{black!10} \cellcolor{white}  & \cellcolor{white} \multirow{-2}{*}{$n=100$} &  T & 4.94 & 0.12 & 2.62 & 23 & 7.50 & 9.84 & 1.75 &  -  & 1.00 &  - \\
 \cellcolor{white}  & \cellcolor{white}  &  L & 1.83 & 7.41 & 0.12 & 0 & 17.30 & 4.83 & 1.14 & 0.18 &  -  & 0.07\\
 \rowcolor{black!10} \cellcolor{white}  & \cellcolor{white} \multirow{-2}{*}{$n=200$} &  T & 1.89 & 0.06 & 0.21 & 78 & 9.85 & 4.66 & 1.14 &  -  & 1.22 &  - \\
 \cellcolor{white}  & \cellcolor{white}  &  L & 0.86 & 4.32 & 0.01 & 2 & 14.31 & 2.53 & 0.69 & 0.18 &  -  & 0.04\\
 \rowcolor{black!10} \cellcolor{white}  & \cellcolor{white} \multirow{-2}{*}{$n=500$} &  T & 0.87 & 0.01 & 0.04 & 96 & 9.97 & 2.48 & 0.69 &  -  & 0.96 &  - \\
 \cellcolor{white}  & \cellcolor{white}  &  L & 0.55 & 3.27 & 0.00 & 8 & 13.27 & 1.70 & 0.49 & 0.08 &  -  & 0.03\\
 \rowcolor{black!10} \cellcolor{white} \multirow{-10}{*}{$\tau_0= 0.5$} & \cellcolor{white} \multirow{-2}{*}{$n=1000$} &  T & 0.55 & 0.01 & 0.01 & 98 & 10.00 & 1.67 & 0.49 &  -  & 0.80 &  - \\
\bottomrule
\end{tabular}
\caption{Lasso (white background) and Thresholded Lasso (grey background). Increasing sample size with $m=100$ and 2 locations of $\tau_0$.}
\label{tab:xpsmpl}
\end{table}

As expected, the probability of correct model selection tends to one for the thresholded scaled Lasso. For the plain scaled Lasso, on the other hand, this probability reaches at most 11\%. As seen already in Table \ref{tab:xpzero}, the problem that the scaled Lasso suffers from is false positives -- it fails to exclude irrelevant variables even as the sample size increases. Finally, and as expected, the penalty applied ($\lambda$) decreases as $n$ increases.

Table \ref{tab:xppsize} considers different values of the non-zero coefficients to investigate the effect of the scale of these coefficients. The data is generated as:
\begin{itemize}
	\item Sample size: $n=100,200$.
	\item $\beta = a[1,1,1,1,1,0,...,0]$, $\delta = a[1,-1,1,-1,1,0,...,0]$.
	\item $a=0.3,0.5,1,2$ is the scale of the non zero parameters. 
\end{itemize}

\begin{table}[!htpb]
\centering
\begin{tabular}{ l l H r r r r H r r r r r }
\toprule
 &  &  & \rotpar{ MSE } & \rotpar{ {\small $ |J(\hat \alpha) \cap J(\alpha_0)^c|$} } & \rotpar{ {\small $|J(\alpha_0) \cap J(\hat\alpha)^c|$} } & \rotpar{ Perfect Sel } & \rotpar{ \# non Zeros } & \rotpar{ $\enVert{\hat\alpha - \alpha_0}_1$ } & \rotpar{ $\enVert{\hat\alpha - \alpha_0}_{\infty}$ } & \rotpar{ $|\hat\tau - \tau_0|$ } & C & $\hat\lambda$\\
\midrule
 \cellcolor{white}  & \cellcolor{white}  &  L & 0.50 & 0.41 & 5.50 & 0 & 4.91 & 2.27 & 0.30 & 0.28 &  -  & 0.15\\
 \rowcolor{black!10} \cellcolor{white}  & \cellcolor{white} \multirow{-2}{*}{$n=100$} &  T & 0.52 & 0.02 & 6.22 & 0 & 3.80 & 2.30 & 0.30 &  -  & 0.46 &  - \\
 \cellcolor{white}  & \cellcolor{white}  &  L & 0.38 & 0.29 & 4.00 & 0 & 6.29 & 1.89 & 0.30 & 0.32 &  -  & 0.10\\
 \rowcolor{black!10} \cellcolor{white}  & \cellcolor{white} \multirow{-2}{*}{$n=200$} &  T & 0.39 & 0.01 & 4.56 & 0 & 5.45 & 1.91 & 0.30 &  -  & 0.44 &  - \\
 \cellcolor{white}  & \cellcolor{white}  &  L & 0.31 & 0.60 & 1.74 & 1 & 8.86 & 1.38 & 0.30 & 0.10 &  -  & 0.04\\
 \rowcolor{black!10} \cellcolor{white} \multirow{-6}{*}{$a=0.3$} & \cellcolor{white} \multirow{-2}{*}{$n=1000$} &  T & 0.31 & 0.00 & 2.21 & 5 & 7.80 & 1.38 & 0.30 &  -  & 0.51 &  - \\
\cmidrule(l){1 - 13}
 \cellcolor{white}  & \cellcolor{white}  &  L & 0.75 & 0.57 & 4.49 & 0 & 6.09 & 3.43 & 0.50 & 0.25 &  -  & 0.15\\
 \rowcolor{black!10} \cellcolor{white}  & \cellcolor{white} \multirow{-2}{*}{$n=100$} &  T & 0.78 & 0.03 & 5.15 & 0 & 4.88 & 3.45 & 0.50 &  -  & 0.47 &  - \\
 \cellcolor{white}  & \cellcolor{white}  &  L & 0.57 & 0.50 & 3.21 & 0 & 7.28 & 2.92 & 0.50 & 0.27 &  -  & 0.10\\
 \rowcolor{black!10} \cellcolor{white}  & \cellcolor{white} \multirow{-2}{*}{$n=200$} &  T & 0.58 & 0.01 & 3.95 & 0 & 6.06 & 2.93 & 0.50 &  -  & 0.48 &  - \\
 \cellcolor{white}  & \cellcolor{white}  &  L & 0.31 & 2.75 & 0.04 & 9 & 12.71 & 1.37 & 0.32 & 0.10 &  -  & 0.03\\
 \rowcolor{black!10} \cellcolor{white} \multirow{-6}{*}{$a=0.5$} & \cellcolor{white} \multirow{-2}{*}{$n=1000$} &  T & 0.31 & 0.00 & 0.06 & 94 & 9.94 & 1.35 & 0.32 &  -  & 0.75 &  - \\
\cmidrule(l){1 - 13}
 \cellcolor{white}  & \cellcolor{white}  &  L & 1.87 & 1.12 & 3.52 & 0 & 7.60 & 6.31 & 1.00 & 0.22 &  -  & 0.18\\
 \rowcolor{black!10} \cellcolor{white}  & \cellcolor{white} \multirow{-2}{*}{$n=100$} &  T & 1.94 & 0.05 & 4.21 & 0 & 5.84 & 6.31 & 1.00 &  -  & 0.56 &  - \\
 \cellcolor{white}  & \cellcolor{white}  &  L & 1.09 & 3.95 & 1.16 & 0 & 12.79 & 4.46 & 0.86 & 0.21 &  -  & 0.09\\
 \rowcolor{black!10} \cellcolor{white}  & \cellcolor{white} \multirow{-2}{*}{$n=200$} &  T & 1.12 & 0.04 & 1.54 & 39 & 8.50 & 4.39 & 0.86 &  -  & 0.88 &  - \\
 \cellcolor{white}  & \cellcolor{white}  &  L & 0.34 & 2.98 & 0.00 & 9 & 12.98 & 1.43 & 0.35 & 0.08 &  -  & 0.03\\
 \rowcolor{black!10} \cellcolor{white} \multirow{-6}{*}{$a=1$} & \cellcolor{white} \multirow{-2}{*}{$n=1000$} &  T & 0.34 & 0.00 & 0.01 & 99 & 10.00 & 1.41 & 0.35 &  -  & 0.83 &  - \\
\cmidrule(l){1 - 13}
 \cellcolor{white}  & \cellcolor{white}  &  L & 4.68 & 5.32 & 2.12 & 0 & 13.20 & 10.01 & 1.76 & 0.20 &  -  & 0.21\\
 \rowcolor{black!10} \cellcolor{white}  & \cellcolor{white} \multirow{-2}{*}{$n=100$} &  T & 4.89 & 0.10 & 2.61 & 21 & 7.49 & 9.80 & 1.76 &  -  & 1.02 &  - \\
 \cellcolor{white}  & \cellcolor{white}  &  L & 1.81 & 7.44 & 0.11 & 0 & 17.33 & 4.74 & 1.12 & 0.18 &  -  & 0.07\\
 \rowcolor{black!10} \cellcolor{white}  & \cellcolor{white} \multirow{-2}{*}{$n=200$} &  T & 1.87 & 0.05 & 0.21 & 78 & 9.84 & 4.57 & 1.12 &  -  & 1.23 &  - \\
 \cellcolor{white}  & \cellcolor{white}  &  L & 0.56 & 3.18 & 0.00 & 7 & 13.18 & 1.70 & 0.49 & 0.07 &  -  & 0.03\\
 \rowcolor{black!10} \cellcolor{white} \multirow{-6}{*}{$a=2$} & \cellcolor{white} \multirow{-2}{*}{$n=1000$} &  T & 0.56 & 0.01 & 0.01 & 98 & 10.00 & 1.68 & 0.49 &  -  & 0.79 &  - \\
\bottomrule
\end{tabular}
\caption{ Lasso (white background) and Thresholded Lasso (grey background). Increasing parameter scale, 3 sample sizes, $\tau_0 = 0.5$.}
\label{tab:xppsize}
\end{table}

When these are as small as $0.3$ perfect model selection does not seem possible unless when $n=1000$. On the other hand, the number of relevant variables excluded clearly decreases as $n$ is increased. In general, no matter what the value of the non-zero coefficients are, all performance measures improve as $n$ is increased, thus confirming the findings in Table \ref{tab:xpsmpl}. While variable selection is easier when the non-zero coefficients are well-separated from the zero ones, the MSE and estimation error of $\hat{\alpha}$ actually improve as the non-zero coefficients become smaller. The reason for this is that falsely classifying a non-zero coefficient as zero is less costly in terms of estimation error when this coefficient is already close to zero than when it is far from zero. On the other hand, $\hat{\tau}$ is estimated slightly more precisely as the non-zero coefficients become more separated from the zero ones.

\begin{table}[!htpb]
\centering
\begin{tabular}{ l l H r r r r H r r r r r }
\toprule
 &  &  & \rotpar{ MSE } & \rotpar{ {\small $ |J(\hat \alpha) \cap J(\alpha_0)^c|$} } & \rotpar{ {\small $|J(\alpha_0) \cap J(\hat\alpha)^c|$} } & \rotpar{ Perfect Sel } & \rotpar{ \# non Zeros } & \rotpar{ $\enVert{\hat\alpha - \alpha_0}_1$ } & \rotpar{ $\enVert{\hat\alpha - \alpha_0}_{\infty}$ } & \rotpar{ $|\hat\tau - \tau_0|$ } & C & $\hat\lambda$\\
\midrule
 \cellcolor{white}  & \cellcolor{white}  &  L & 0.34 & 0.10 & 0.00 & 90 & 2.10 & 0.36 & 0.25 & 0.25 &  -  & 0.09\\
 \rowcolor{black!10} \cellcolor{white}  & \cellcolor{white} \multirow{-2}{*}{$\tau_0=0.3$} &  T & 0.34 & 0.01 & 0.00 & 99 & 2.00 & 0.35 & 0.25 &  -  & 0.23 &  - \\
 \cellcolor{white}  & \cellcolor{white}  &  L & 0.35 & 0.09 & 0.00 & 91 & 2.09 & 0.36 & 0.26 & 0.22 &  -  & 0.09\\
 \rowcolor{black!10} \cellcolor{white}  & \cellcolor{white} \multirow{-2}{*}{$\tau_0=0.4$} &  T & 0.35 & 0.01 & 0.00 & 99 & 2.00 & 0.36 & 0.26 &  -  & 0.20 &  - \\
 \cellcolor{white}  & \cellcolor{white}  &  L & 0.36 & 0.10 & 0.00 & 90 & 2.10 & 0.37 & 0.27 & 0.23 &  -  & 0.09\\
 \rowcolor{black!10} \cellcolor{white} \multirow{-6}{*}{$m_1=1$} & \cellcolor{white} \multirow{-2}{*}{$\tau_0=0.5$} &  T & 0.36 & 0.01 & 0.00 & 99 & 2.01 & 0.37 & 0.27 &  -  & 0.19 &  - \\
\cmidrule(l){1 - 13}
 \cellcolor{white}  & \cellcolor{white}  &  L & 1.84 & 1.11 & 0.19 & 31 & 10.93 & 2.59 & 0.90 & 0.19 &  -  & 0.08\\
 \rowcolor{black!10} \cellcolor{white}  & \cellcolor{white} \multirow{-2}{*}{$\tau_0=0.3$} &  T & 1.86 & 0.01 & 0.28 & 78 & 9.73 & 2.57 & 0.91 &  -  & 0.68 &  - \\
 \cellcolor{white}  & \cellcolor{white}  &  L & 2.03 & 1.11 & 0.20 & 32 & 10.91 & 2.67 & 0.93 & 0.18 &  -  & 0.08\\
 \rowcolor{black!10} \cellcolor{white}  & \cellcolor{white} \multirow{-2}{*}{$\tau_0=0.4$} &  T & 2.05 & 0.02 & 0.30 & 76 & 9.72 & 2.64 & 0.93 &  -  & 0.63 &  - \\
 \cellcolor{white}  & \cellcolor{white}  &  L & 2.05 & 1.01 & 0.16 & 35 & 10.85 & 2.57 & 0.91 & 0.17 &  -  & 0.08\\
 \rowcolor{black!10} \cellcolor{white} \multirow{-6}{*}{$m_1=5$} & \cellcolor{white} \multirow{-2}{*}{$\tau_0=0.5$} &  T & 2.06 & 0.02 & 0.27 & 79 & 9.74 & 2.55 & 0.91 &  -  & 0.61 &  - \\
\cmidrule(l){1 - 13}
 \cellcolor{white}  & \cellcolor{white}  &  L & 5.08 & 2.85 & 0.81 & 5 & 22.04 & 6.54 & 1.36 & 0.19 &  -  & 0.08\\
 \rowcolor{black!10} \cellcolor{white}  & \cellcolor{white} \multirow{-2}{*}{$\tau_0=0.3$} &  T & 5.12 & 0.06 & 1.06 & 51 & 19.01 & 6.48 & 1.36 &  -  & 1.09 &  - \\
 \cellcolor{white}  & \cellcolor{white}  &  L & 4.84 & 2.68 & 0.66 & 7 & 22.01 & 6.17 & 1.28 & 0.18 &  -  & 0.08\\
 \rowcolor{black!10} \cellcolor{white}  & \cellcolor{white} \multirow{-2}{*}{$\tau_0=0.4$} &  T & 4.88 & 0.05 & 0.89 & 57 & 19.16 & 6.11 & 1.28 &  -  & 1.01 &  - \\
 \cellcolor{white}  & \cellcolor{white}  &  L & 5.05 & 2.56 & 0.65 & 7 & 21.91 & 6.10 & 1.25 & 0.18 &  -  & 0.08\\
 \rowcolor{black!10} \cellcolor{white} \multirow{-6}{*}{$m_1=10$} & \cellcolor{white} \multirow{-2}{*}{$\tau_0=0.5$} &  T & 5.09 & 0.04 & 0.90 & 59 & 19.14 & 6.05 & 1.25 &  -  & 0.95 &  - \\
\cmidrule(l){1 - 13}
 \cellcolor{white}  & \cellcolor{white}  &  L & 19.93 & 9.92 & 4.35 & 0 & 55.58 & 23.72 & 1.87 & 0.20 &  -  & 0.07\\
 \rowcolor{black!10} \cellcolor{white}  & \cellcolor{white} \multirow{-2}{*}{$\tau_0=0.3$} &  T & 20.27 & 0.31 & 5.56 & 10 & 44.75 & 23.45 & 1.88 &  -  & 2.45 &  - \\
 \cellcolor{white}  & \cellcolor{white}  &  L & 19.13 & 10.32 & 3.56 & 0 & 56.76 & 22.76 & 1.88 & 0.20 &  -  & 0.07\\
 \rowcolor{black!10} \cellcolor{white}  & \cellcolor{white} \multirow{-2}{*}{$\tau_0=0.4$} &  T & 19.48 & 0.30 & 4.63 & 11 & 45.67 & 22.46 & 1.88 &  -  & 2.31 &  - \\
 \cellcolor{white}  & \cellcolor{white}  &  L & 18.32 & 9.90 & 3.05 & 0 & 56.85 & 21.53 & 1.75 & 0.23 &  -  & 0.07\\
 \rowcolor{black!10} \cellcolor{white} \multirow{-6}{*}{$m_1=25$} & \cellcolor{white} \multirow{-2}{*}{$\tau_0=0.5$} &  T & 18.62 & 0.30 & 3.94 & 19 & 46.37 & 21.25 & 1.75 &  -  & 2.09 &  - \\
\bottomrule
\end{tabular}
\caption{ Lasso (white background) and Thresholded Lasso (grey background). Increasing number of non zero parameters ($m_1$), fixed number of zeros ($m_0=100$), and 3 locations of $\tau_0$.}
\label{tab:xpones}
\end{table}

Finally, Table \ref{tab:xpones} investigates the effect of reducing the sparsity of the model, i.e. of increasing the number of non-zero coefficients.
\begin{itemize}
	\item Sample size: $n=200$.
	\item $\beta = [2,...,2,0,...,0]$, $\delta = [2,...,2,0,...,0]$.
	\item $\beta$ and $\delta$ contain both $m_0=100$ parameters equal to zero.
	\item $\beta$ and $\delta$ contain both $m_1=1,5,10,50$ parameters equal to 2.
	\item The length $\beta$ and $\delta$ is $m=m_0+m_1$.
\end{itemize}

Irrespective of the value of $\tau_0$, perfect model selection becomes harder as the number of relevant variables increases. As our theory is based on the assumption of sparsity, this is not a surprising finding. The MSE and estimation error of $\hat{\alpha}$ also increase by a lot while the estimation error of $\hat{\tau}$ is virtually unaffected by the number of relevant variables. Notice that the threshold parameter, $C$, increases drastically as the number of non-zero coefficients increases. The explanation for this is that thresholding seeks to avoid excluding one of the many relevant variables by setting the threshold higher as there are now more relevant variables at risk of being exluded.

\section{Application}\label{application}

This application aims at investigating the presence of a threshold in the effect of debt on future GDP growth. The academic discussion regarding the impact of debt on growth, and the existence of a threshold above which debt becomes severely detrimental to future growth, has been reignited by \cite{RR} who provided evidence for the existence of such a threshold. The evidences presented by \cite{RR} have been challenged by \cite{herndon2014does}, but others have put forth supportive evidences for this thesis, see among others \cite{BIS,fritzi2010,baum2013debt}. 
\subsection{Data}
We use the data made available by \cite{BIS}\footnote{The original data is available at \url{http://www.bis.org/publ/work352.htm}, and can also be found in the replication material for this section.} which originates mainly from the IMF and OECD data bases. The data contains four measures of debt-to-GDP ratio for:
\begin{enumerate}
  \item Government debt,
  \item Corporate debt,
  \item Private debt (corporate + household),
  \item Total (non financial institutions) debt (private + government).
\end{enumerate}
Notice that private and total debt are aggregate measures of debt.

The data of \cite{BIS} also contains a measure of household debt that we drop as the series is incomplete.
A set of control variables, composed of standard macroeconomic indicators, is also included in the data.
\begin{enumerate}
  \item GDP: The logarithm of the \textit{per capita} GDP.
  \item Savings: Gross savings to GDP ratio.
  \item $\Delta$Pop: Population growth.
  \item School: Years spent in secondary education.
  \item Open: Openness to trade, exports plus imports over GDP.
  \item $\Delta$CPI: Inflation.
  \item Dep: Population dependency ratio.
  \item LL: Ratio of liquid liabilities to GDP.
  \item Crisis: An indicator for banking crisis in the subsequent 5 years. This is taken from \cite{RR}.
\end{enumerate}

The data is observed for 18 countries\footnote{US, Japan, Germany, the United Kingdom, France, Italy, Canada, Australia, Austria, Belgium, Denmark, Finland, Greece, the Netherlands, Norway, Portugal, Spain, and Sweden.} from 1980 to 2009 at an annual frequency. We lose one observation at the start of the sample due to first differencing and five at the end of the sample due to computing the 5 years ahead average growth rate, so that the full sample is 1981-2004. The details on the construction of each variables can be found in \cite{BIS}.

\subsection{Results}

In order to evaluate the impact of debt on growth, as well as the potential presence of a threshold in this effect, we estimate a set of growth regressions. As in \cite{BIS} our left hand side variable is the 5 years forward average rate of growth of per capita GDP. Even though our estimator is not a panel estimator we choose to pool the data so as to make our results comparable with those of \cite{BIS} and benefit from a larger sample. 

\begin{table}
  \begin{tabular}{ l l r r r r r r r r }
    \toprule
    & Threshold:& \multicolumn{2}{c}{ Government }& \multicolumn{2}{c}{ Government }& \multicolumn{2}{c}{ Government }& \multicolumn{2}{c}{ Government }\\
    & & L & T&L & T&L & T&L & T\\
    \midrule
    \multirow{ 11 }{*}{$\hat{\beta}$} & intercept & 42.43 & 42.43 & 79.611 & 79.611 & 86.416 & 86.416 & 136.988 & 136.988\\
    & GDP & -3.643 & -3.643 & -7.419 & -7.419 & -7.495 & -7.495 & -11.621 & -11.621\\
    & Savings & -0.035 & -0.035 & 0.033 & 0.033 & 0.02 & 0.02 &  & \\
    & $\Delta$Pop & -1.692 & -1.692 & -1.493 & -1.493 & -0.879 & -0.879 & -0.813 & -0.813\\
    & School & 0.426 & 0.426 & 0.507 & 0.507 & 0.095 & 0.095 & -0.082 & -0.082\\
    & Open & 0.003 &  & 0.026 &  & 0.024 & 0.024 & 0.037 & 0.037\\
    & $\Delta$CPI & -0.061 & -0.061 & -0.056 & -0.056 & -0.157 & -0.157 & -0.252 & -0.252\\
    & Dep & -0.091 & -0.091 & -0.104 & -0.104 & -0.132 & -0.132 & -0.22 & -0.22\\
    & LL & -0.433 & -0.433 & 0.33 & 0.33 & 0.574 & 0.574 & 0.631 & 0.631\\
    & Crisis & -1.277 & -1.277 & -1.58 & -1.58 & -0.949 & -0.949 & -1.396 & -1.396\\
    & Government & -0.713 & -0.713 &  &  &  &  & -0.518 & -0.518\\

    \cmidrule(lr){1 - 10}

    \multirow{ 11 }{*}{$\hat{\delta}$} & intercept & -12.167 & -12.167 & -1.504 & -1.504 &  &  &  & \\
    & GDP &  &  &  &  &  &  &  & \\
    & Savings & 0.087 & 0.087 & -0.037 &  & -0.052 & -0.052 & 0.008 & \\
    & $\Delta$Pop & 1.563 & 1.563 & 0.42 & 0.42 & 0.222 & 0.222 & 0.61 & 0.61\\
    & School & -0.077 & -0.077 &  &  & 0.203 & 0.203 & 0.098 & 0.098\\
    & Open & -0.006 &  & 0.007 &  & 0.012 &  &  & \\
    & $\Delta$CPI &  &  &  &  &  &  &  & \\
    & Dep & 0.181 & 0.181 &  &  & -0.035 & -0.035 &  & \\
    & LL & 0.827 & 0.827 & 0.909 & 0.909 &  &  &  & \\
    & Crisis & -0.459 & -0.459 & -0.294 & -0.294 & -1.338 & -1.338 &  & \\
    & Government & 1.762 & 1.762 & 1.471 & 1.471 &  &  & -3.23 & -3.23\\
    \midrule
    & $\widehat{\tau}$ & 0.82&0.82&0.68&0.68&0.59&0.59&0.65&0.65\\
    & $\widehat{\lambda}$ & 0.007&0.007&0.015&0.015&0.007&0.007&0.008&0.008\\
    & $\widehat{C}$ &  - &  0.1& - &  0.3& - &  0.1& - &  0.1\\
    & Sample & \multicolumn{2}{c}{1981 - 2004} & \multicolumn{2}{c}{1981 - 2004} & \multicolumn{2}{c}{1990 - 2004} & \multicolumn{2}{c}{No overlap}  \\
    & FE & \multicolumn{2}{c}{$\times$} & \multicolumn{2}{c}{\checkmark} & \multicolumn{2}{c}{\checkmark} & \multicolumn{2}{c}{\checkmark}  \\
    \bottomrule
  \end{tabular}\caption{4 specifications with government debt included as threshold variable and regressor. Estimated parameters for the Lasso (L) and Thresholded Lasso (T). Empty cells are parameters set to zero, dashes indicate parameters not included in the model.}\label{tab:panelGvt}\end{table}

  We report a first set of results focusing on the impact of government debt on future GDP growth in Table \ref{tab:panelGvt}. We consider 3 different samples: 1981 to 2004 (full sample, 414 observations), 1990 to 2004 (252 observations), and a sample with no overlapping data (5 years\footnote{1984,1989,1994,1999,2004.}, 90 observations). For the full sample we report results for models estimated with and without country specific dummies (noted FE in the tables). We do not report the estimated parameters associated with the country specific dummies.

  We estimate the models including every control variable and a single debt measure, that is, 23 parameters to estimate (11 parameters in $\beta$,11 parameters in $\delta$, and the threshold parameter $\tau$) including the intercept and the thresholded intercept plus, in some instances, 17 country specific dummies. The country specific dummies are not penalized. The grid of threshold parameters goes from the $15^{th}$ to the $85^{th}$ centiles of the threshold variable by steps of 5 centiles. We select the thresholding parameter $C$ by BIC using a grid from $0.1$ to $5$, so that parameters smaller (in absolute value) than $\widehat C \widehat \lambda$ are set to zero by the thresholded scaled Lasso.

  Table \ref{tab:panelGvt} reports the estimated parameters for the 4 specifications of the model, all including government debt.  The $L$ and $T$ in the header of the table indicates a scaled Lasso estimate ($\widehat \beta$, $\widehat \delta$) or thresholded scaled Lasso estimate ($\widetilde \beta$, $\widetilde \delta$). The upper panel of each table reports $\widehat \beta$ and $\widetilde \beta$, the middle panel $\widehat \delta$ and $\widetilde \delta$, and the lower panel give the values of $\widehat \tau$, $\widehat \lambda$, and $\widehat C$. Recall that the effect of the regressors when the threshold variable is below its threshold is given by $\widehat \beta + \widehat \delta$ ($\widetilde \beta + \widetilde \delta$) while the effect when the threshold variable is above its threshold is given by $\widehat \beta$ ($\widetilde \beta$) for the scaled Lasso (thresholded scaled Lasso).

  \begin{table}
  \centering
    \begin{tabular}{ l l r r r r r r }
      \toprule
      & Threshold:& \multicolumn{2}{c}{ Corporate }& \multicolumn{2}{c}{ Private }& \multicolumn{2}{c}{ Total }\\
      & & L & T&L & T&L & T\\
      \midrule
      \multirow{ 12 }{*}{$\hat{\beta}$} & intercept & 140.097 & 140.097 & 126.236 & 126.236 & 134.725 & 134.725\\
      & GDP & -11.642 & -11.642 & -10.616 & -10.616 & -11.396 & -11.396\\
      & Savings & -0.026 & -0.026 & -0.031 & -0.031 & -0.011 & -0.011\\
      & $\Delta$Pop & -1.063 & -1.063 &  &  & -0.995 & -0.995\\
      & School & -0.172 & -0.172 &  &  & -0.132 & -0.132\\
      & Open & 0.053 & 0.053 & 0.041 & 0.041 & 0.047 & 0.047\\
      & $\Delta$CPI & -0.204 & -0.204 & -0.19 & -0.19 & -0.166 & -0.166\\
      & Dep & -0.242 & -0.242 & -0.191 & -0.191 & -0.235 & -0.235\\
      & LL & 0.332 & 0.332 & 0.316 & 0.316 & 0.376 & 0.376\\
      & Crisis & -0.96 & -0.96 & -0.319 & -0.319 & -0.943 & -0.943\\
      &Corporate&0.491&0.491&-&-&-&-\\
      &Private&-&-&-0.968&-0.968&-&-\\
      &Total&-&-&-&-&0.284&0.284\\

      \cmidrule(lr){1 - 8}

      \multirow{ 12 }{*}{$\hat{\delta}$} & intercept & 8.261 & 8.261 & 2.301 & 2.301 &  & \\
      & GDP &  &  &  &  &  & \\
      & Savings & -0.243 & -0.243 & 0.022 & 0.022 &  & \\
      & $\Delta$Pop & -2.154 & -2.154 & -1.1 & -1.1 & 2.387 & 2.387\\
      & School & -0.29 & -0.29 & -0.33 & -0.33 & 0.387 & 0.387\\
      & Open &  &  & -0.007 &  & 0.063 & 0.063\\
      & $\Delta$CPI & -0.032 & -0.032 & -0.082 & -0.082 & 0.777 & 0.777\\
      & Dep &  &  &  &  & -0.192 & -0.192\\
      & LL & 1.175 & 1.175 & 0.365 & 0.365 &  & \\
      & Crisis & -2.389 & -2.389 & -1.167 & -1.167 & -31.521 & -31.521\\
      & Corporate&&&-&-&-&-\\
      & Private&-&-&0.563&0.563&-&-\\
      & Total&-&-&-&-&&\\
      \midrule
      & $\widehat{\tau}$ & 0.69&0.69&1.62&1.62&2&2\\
      & $\widehat{\lambda}$ & 0.001&0.001&0.005&0.005&0.002&0.002\\
      & $\widehat{C}$ &  - &  0.1& - &  0.1& - &  0.1\\
      & Sample & \multicolumn{2}{c}{1981 - 2004} & \multicolumn{2}{c}{1981 - 2004} & \multicolumn{2}{c}{1981 - 2004}  \\
      & FE & \multicolumn{2}{c}{\checkmark} & \multicolumn{2}{c}{\checkmark} & \multicolumn{2}{c}{\checkmark}  \\
      \bottomrule
    \end{tabular}\caption{Growth regressions with corporate, private, or total debt (see header) included both as threshold variable and as regressor. Estimated parameters, pooled data, Lasso (L) and Thresholded Lasso (T). Empty cells are parameters set to zero, dashes indicate parameters not included in the model.}\label{tab:panelDbt}\end{table}

    A large fraction of $\widehat \beta$ is non zero, the Lasso drops a single variable twice, while $\widehat \delta$ is more sparse, the Lasso drops between 2 and 7 variables. The thresholding parameter $\widehat C$ is always chosen among the lowest values in the search grid, this nonetheless results in between 1 and 3 extra parameters being discarded compared to the scaled Lasso.
    A threshold ($\widehat \tau$) for the effect of government debt on growth is found at between 60\% and 80\% of GDP, consistent with the findings of \cite{BIS,RR,fritzi2010,baum2013debt}.

    The level of GDP is found to have a negative effect on GDP per capita growth as predicted by the income convergence hypothesis, as do inflation, the dependency ratio, population growth, and crises.
    Considering the effect of both $\widehat \beta$ and $\widehat \delta$, our model indicates in most instances that government debt has a positive effect below the threshold and a negative effect, or no effect at all, above the debt threshold. \textit{Ceteris paribus} a 10 percentage point increase in the government debt to GDP ratio, when it is above the threshold, is found to result in a decrease of the average 5 year growth rate between 0.07\% and zero. Looking at this effect of high debt on future growth in isolation is overly restrictive though since there are large changes in the other parameters of the model when the debt threshold is crossed. This is in particular the case for financial variables. Interestingly, crises are found to have a more detrimental effect on growth for countries with a government debt ratio below the threshold and while liquid liabilities (LL) are beneficial to the future growth of a country with low debt this does not appear to be the case when debt is high.

    Table \ref{tab:panelDbt} reports estimates for 3 other measures of debt in a model with country dummies and using the full sample, the same model used in the first two columns of Table \ref{tab:panelGvt}. The sparsity pattern in Table \ref{tab:panelDbt} is comparable to that of Table \ref{tab:panelGvt} and some similarities are found between the estimated values. Again, the level of \textit{per capita} GDP is found to have a negative impact on future growth, as are the dependency ratio, inflation, population growth, and financial crisis.

A threshold is always found and identified, 69\% for corporate debt, 162\% for private debt, and 200\% for the total debt. The large value of the estimated thresholds for private and total debt can be explained by the fact that these are aggregate measures of debt and hence of a substantially larger magnitude than either corporate of government debts. The effect of corporate and total debt is found to be positive and not directly affected by the threshold whereas the effect of private debt is negative, and more so when private debt is high. As previously, financial crises are found to have a stronger negative impact on countries with low debt, though crises are detrimental to growth irrespective of the level of debt.

\section{Conclusion}
In this paper we considered high-dimensional threshold regressions and provided sup-norm oracle inequalities for the estimation error of the scaled Lasso of \cite{lee2012lasso}. These results are non-trivial as most research has focused on either $\ell_1$ or $\ell_2$ oracle inequalities. The sup-norm bounds are shown to be crucial for exact variable selection by means of thresholding. To be precise, we can distinguish at a much finer scale between zero and non-zero coefficients than would have been possible if thresholding had been based on either $\ell_1$ or $\ell_2$ oracle inequalities.

We carry out simulations and show that the thresholded scaled Lasso performs well in model selection. Finally, we estimate a set of growth regression documenting the existence of a threshold in the amount of debt relative to GDP. Several parameters change when the threshold is crossed making the effect of high debt on future growth unclear.

Future work includes investigating the effect of multiple thresholds.

\vspace{2in}
\begin{center}
{\bf APPENDIX}
\end{center}

The following result is needed in the proofs of Theorems \ref{thm1} and \ref{thm2}. It is similar to Lemma 6 in \cite{lee2012lasso} but allows for random regressors and non-gaussian error terms.
\begin{lemma}\label{Xu}
Let Assumption 1 be satisfied. Then,
\begin{align*}
\enVert[2]{\frac{1}{n}X'(\hat{\tau})U}_{\ell_\infty}=O_p\del[3]{\sqrt{\frac{\log(m)}{n}}}
\end{align*}  
\end{lemma}

\begin{proof}
First, note that $\enVert[2]{\frac{1}{n}X'(\hat{\tau})U}_{\ell_\infty}\leq \sup_{\tau\in T}\enVert[2]{\frac{1}{n}X'(\tau)U}_{\ell_\infty}$ such that it suffices to bound the right hand side.
Let $\epsilon>0$ be arbitrary. By the independence of $(X_1,...,X_n, U_1,...,U_n)$ and $(Q_1,...,Q_n)$ one has for $j=1,...,m$,
\begin{align}
P\del[3]{\sup_{\tau\in T}\envert[2]{\frac{1}{n}\sum_{i=1}^nX_i^{(j)}U_i1_{\cbr[0]{Q_i<\tau}}}>\epsilon \sVert[1](Q_1,...,Q_n)}
&=
P\del[3]{\max_{1\leq k\leq n}\envert[2]{\frac{1}{n}\sum_{i=1}^kX_i^{(j)}U_i}>\epsilon \sVert[1](Q_1,...,Q_n)}\notag\\
&=
P\del[3]{\max_{1\leq k\leq n}\envert[2]{\frac{1}{n}\sum_{i=1}^kX_i^{(j)}U_i}>\epsilon}\label{first}
\end{align}
almost surely, where the first equality used that conditional on $(Q_1,...,Q_n)$, $\del[1]{1_{\cbr[0]{Q_1<\tau}},...,1_{\cbr[0]{Q_n<\tau}}}$ can only take $n$ different values (and sorted $\cbr[0]{X_i, U_i, Q_i}_{i=1}^n$ by $(Q_1,...,Q_n)$ in ascending order). The second equality used the independence $(X_1,...,X_n, U_1,...,U_n)$ and $(Q_1,...,Q_n)$. Next, by Corollary 4 in \cite{montgomery1993comparison} 
there exists a universal constant $c>0$ such that
\begin{align}
P\del[3]{\max_{1\leq k\leq n}\envert[2]{\frac{1}{n}\sum_{i=1}^kX_i^{(j)}U_i}>\epsilon}
\leq
cP\del[3]{\envert[2]{\sum_{i=1}^nX_i^{(j)}U_i}>\frac{\epsilon n}{c}}
\end{align}
As $X_i^{(j)}U_i$ is subexponential (the product of two subgaussian variables is subexponential) for all $i=1,...,n$ and $j=1,...,m$, Corollary 5.17 in \cite{vershynin2010introduction} yields
\begin{align}
P\del[3]{\envert[2]{\sum_{i=1}^nX_i^{(j)}U_i}>\frac{\epsilon n}{c}}
\leq
2\exp\del[2]{-d\sbr[1]{(\epsilon/K)^2\wedge (\epsilon/K)}n}
\label{vershynin}
\end{align}
where $d>0$ and $K=K(c)>0$ are absolute constants. Therefore, choosing $\epsilon=A\sqrt{\frac{\log(m)}{n}}$ for some $A\geq 1$ yields
\begin{align*}
P\del[3]{\envert[2]{\sum_{i=1}^nX_i^{(j)}U_i}>\frac{\epsilon n}{c}}
\leq
2\exp\del[3]{-\frac{dA}{K^2\vee K}\sbr[3]{\frac{\log(m)}{n}\wedge \sqrt{\frac{\log(m)}{n}}}n}
\leq
2\exp\del[3]{-\frac{dA}{K^2\vee K}\log(m)}
\end{align*}
where the second estimate used that $\log(m)/n\to 0$ such that $\frac{\log(m)}{n}$ is smaller than its square root for $n$ sufficiently large. Hence,
\begin{align*}
P\del[3]{\sup_{\tau\in T}\envert[2]{\frac{1}{n}\sum_{i=1}^nX_i^{(j)}U_i1_{\cbr[0]{Q_i<\tau}}}>\epsilon \sVert[1](Q_1,...,Q_n)}
\leq
2c\exp\del[3]{-\frac{dA}{K^2\vee K}\log(m)}
\end{align*}
for all $j=1,...,m$ almost surely. Taking expectations over $(Q_1,...,Q_n)$ yields
\begin{align}
P\del[3]{\sup_{\tau\in T}\envert[2]{\frac{1}{n}\sum_{i=1}^nX_i^{(j)}U_i1_{\cbr[0]{Q_i<\tau}}}>\epsilon}
\leq
2c\exp\del[3]{-\frac{dA}{K^2\vee K}\log(m)}\label{bound}.
\end{align}
Therefore, combining (\ref{vershynin}) and (\ref{bound}), a union bound over $2m$ terms yields
\begin{align*}
P\del[3]{\sup_{\tau\in T}\enVert[2]{\frac{1}{n}X'(\tau)U}_{\ell_\infty}>\epsilon}
\leq
2m(1+c)\exp\del[3]{-\frac{dA}{K^2\vee K}\log(m)}.
\end{align*}
Choosing $A$ sufficiently large implies that $\sup_{\tau\in T}\enVert[2]{\frac{1}{n}X'(\tau)U}_{\ell_\infty}=O_p\del[2]{\sqrt{\frac{\log(m)}{n}}}$ using the 
definition of $\epsilon= A \sqrt{\log(m)/n}$.
\end{proof}

\begin{lemma}\label{scaling}
Let assumption 1 be satisfied. Then, $\sup_{\tau\in T}\max_{1\leq j\leq 2m}\enVert[1]{X^{(j)}(\tau)}_n=O_p(1)$ and $\min_{1\leq j\leq 2m}\enVert[1]{X^{(j)}(t_0)}_n$ is bounded away from zero wpa1.
\end{lemma}

\begin{proof}
Consider the first claim and note that $\sup_{\tau\in T}\max_{1\leq j\leq 2m}\enVert[1]{X^{(j)}(\tau)}_n=\max_{1\leq j\leq m}\enVert[1]{X^{(j)}(\tau)}_n$. As $X_1^{(j)}$ is uniformly subgaussian in $j=1,...,m$ it also holds that $E\del[1]{{X_1^{(j)}}^2}$ is uniformly bounded (this follows by Lemma 2.2.1 in \cite{van1996weak} and the inequalities at the bottom of page 95 in that reference). Thus, by the triangle inequality and subadditivity of $x\mapsto \sqrt{x}$,
\begin{align*}
\sqrt{\frac{1}{n}\sum_{i=1}^n{X_i^{(j)}}^2}
\leq
\sqrt{\frac{1}{n}\envert[2]{\sum_{i=1}^n\del[2]{{X_i^{(j)}}^2-E{X_i^{(j)}}^2}}} + \sqrt{E{X_1^{(j)}}^2}
\end{align*}
and hence it suffices to bound $\sqrt{\frac{1}{n}\envert[2]{\sum_{i=1}^n\del[2]{{X_i^{(j)}}^2-E{X_i^{(j)}}^2}}}$, or, equivalently, $\frac{1}{n}\envert[2]{\sum_{i=1}^n\del[2]{{X_i^{(j)}}^2-E{X_i^{(j)}}^2}}$ uniformly in $j=1,...,n$ by a constant with probability tending to 1. As the ${X_i^{(j)}}^2$ are uniformly subexponential (as they are a product of uniformly subgaussian random variables) in $j=1,...,m$ Corollary 5.17 in \cite{vershynin2010introduction} implies that for any $\epsilon>0$ there exist constants $c,K>0$ (see
\cite{vershynin2010introduction} for the exact meaning of the constants) such that
\begin{align*}
P\del[2]{\frac{1}{n}\envert[2]{\sum_{i=1}^n\del[2]{{X_i^{(j)}}^2-E{X_i^{(j)}}^2}}>\epsilon}
\leq
2\exp\del[2]{-c\sbr[1]{(\epsilon/K)^2\wedge (\epsilon/K)}n}
\end{align*}
for all $j=1,...,m$. Now, choosing $\epsilon=K\vee K/c$, the union bound yields that
\begin{align*}
P\del[2]{\max_{1\leq j\leq m}\frac{1}{n}\envert[2]{\sum_{i=1}^n\del[2]{{X_i^{(j)}}^2-E{X_i^{(j)}}^2}}>\epsilon}
\leq 
2me^{-n}
\to 0
\end{align*}
as $\log(m)/n\to 0$. Thus, $K\vee K/c$ is large enough to be the sought constant.

Now turn to the second claim and observe $\min_{1\leq j\leq 2m}\enVert[1]{X^{(j)}(t_0)}_n=\min_{m+1\leq j\leq 2m}\enVert[1]{X^{(j)}(t_0)}_n$. Note that by Assumption 1,
\[\min_{1\leq j\leq m}E\del[1]{{X^{(j)}_{1}}^21_{\cbr[0]{Q_1<t_0}}}=\min_{1\leq j\leq m}E\del[1]{{X^{(j)}_1}^2}t_0=:r>0.\]
where the first equality used the independence of $X_1$ and $Q_1$ as well as that $Q_1$ us uniformly distributed on $[0,1]$. Therefore, it suffices to show that $\max_{1\leq j\leq m}\frac{1}{n}\envert[2]{\sum_{i=1}^n\del[2]{{X_i^{(j)}}^21_{\cbr[0]{Q_i<t_0}}-E{X_i^{(j)}}^21_{\cbr[0]{Q_i<t_0}}}}\leq d\leq r/2$ with probability tending to one. As ${X_1^{(j)}}^21_{\cbr[0]{Q_1<t_0}}$ is subexponential it follows once more from Corollary 5.17 in \cite{vershynin2010introduction} that for $d=K\wedge r/2$
\begin{align*}
P\del[3]{\frac{1}{n}\envert[2]{\sum_{i=1}^n\del[2]{{X_i^{(j)}}^21_{\cbr[0]{Q_i<t_0}}-E{X_i^{(j)}}^21_{\cbr[0]{Q_i<t_0}}}}>d}
\leq
2\exp\del[2]{-c\sbr[1]{(d/K)^2\wedge (d/K)}n}
\leq
2e^{\frac{-cd^2}{K^2}n}
\end{align*}
for $j=1,...,m$. Thus, by the union bound
\begin{align*}
P\del[3]{\max_{1\leq j\leq m}\frac{1}{n}\envert[2]{\sum_{i=1}^n\del[2]{{X_i^{(j)}}^21_{\cbr[0]{Q_i<t_0}}-E{X_i^{(j)}}^21_{\cbr[0]{Q_i<t_0}}}}\geq d}
\leq
2me^{\frac{-cd^2}{K^2}n}
\end{align*}
which tends to zero as $\frac{\log(m)}{n}\to 0$ by assumption 1.
\end{proof}


\begin{proof}[Proof of Theorem \ref{thm1}]
Note first that when $\delta_0=0$, for any random variable $V$
\[ Y_i = X_i' \beta_0 + U_i = X_i'\beta_0+X_i'1_{\cbr{Q_i<V}}\delta_0+U_i,\]
since 
\[X_i'1_{\cbr{Q_i<V}}\delta_0=0.\]

 In particular, this is true for $V=\hat{\tau}$. Next, since $\hat{\alpha}=(\hat{\beta}', \hat{\delta}')'$ satisfies the Karush-Kuhn-Tucker conditions for a minimum, one has
\begin{align*}
-\frac{1}{n}X'(\hat{\tau})\del[1]{Y-X(\hat{\tau})\hat{\alpha}}+\lambda D(\hat{\tau})z(\hat{\tau})
=
0
\end{align*}
where $\enVert[0]{z(\hat{\tau})}_{\ell_\infty}\leq 1$ and $z(\hat{\tau})_j=sign(\hat{\alpha}_j)$ if $\hat{\alpha}_j\neq 0$ and. This can be rewritten as
\begin{align*}
\frac{1}{n}X'(\hat{\tau})X(\hat{\tau})\del[0]{\hat{\alpha}-\alpha_0}
=
\frac{1}{n}X'(\hat{\tau})U_i-\lambda D(\hat{\tau})z(\hat{\tau}).
\end{align*}
which is equivalent to
\begin{align*}
\Sigma(\hat{\tau})\del[0]{\hat{\alpha}-\alpha_0}
=
\del[1]{\Sigma(\hat{\tau})-\frac{1}{n}X'(\hat{\tau})X(\hat{\tau})}\del[0]{\hat{\alpha}-\alpha_0}+\frac{1}{n}X'(\hat{\tau})U-\lambda D(\hat{\tau})z(\hat{\tau}).
\end{align*}
Next, $\Theta(\tau)=\Sigma(\tau)^{-1}$ exists for all $\tau\in T$ under Assumption 1 when $\Sigma$ has full rank as argued in the discussion of Assumption 2. In fact, $\kappa=\kappa(s,3,T)>0$ with probability tending to one as is needed in order to invoke Theorem 2 of \cite{lee2012lasso} below. It follows that $\Sigma(\hat{\tau})$ is invertible with inverse $\Theta(\hat{\tau})$. Thus,
\begin{align*}
\hat{\alpha}-\alpha_0
=
\Theta(\hat{\tau})\del[1]{\Sigma(\hat{\tau})-\frac{1}{n}X'(\hat{\tau})X(\hat{\tau})}\del[0]{\hat{\alpha}-\alpha_0}+\Theta(\hat{\tau})\frac{1}{n}X'(\hat{\tau})U-\lambda \Theta(\hat{\tau})D(\hat{\tau})z(\hat{\tau}).
\end{align*}
Now recall that for matrices $A,B$ and a vector $c$ of compatible dimensions, one has $\enVert[0]{ABc}_{\ell_\infty}\leq \enVert[0]{A}_{\ell_\infty}\enVert[0]{Bc}_{\ell_\infty}\leq \enVert[0]{A}_{\ell_\infty}\enVert[0]{B}_{\infty}\enVert[0]{c}_{\ell_1}$ (see, eg, \cite{horn2013johnson} Chapter 5). Using this,
\begin{align}
\enVert[1]{\hat{\alpha}-\alpha_0}_{\ell_\infty}
&\leq
\enVert[1]{\Theta(\hat{\tau})}_{\ell_\infty}\enVert[2]{\del[1]{\Sigma(\hat{\tau})-\frac{1}{n}X'(\hat{\tau})X(\hat{\tau})}}_{\infty}\enVert[1]{\del[0]{\hat{\alpha}-\alpha_0}}_{\ell_1}\notag\\
&+
\enVert[1]{\Theta(\hat{\tau})}_{\ell_\infty}\enVert[2]{\frac{1}{n}X'(\hat{\tau})U}_{\ell_\infty}+\lambda \enVert[1]{\Theta(\hat{\tau})}_{\ell_\infty}\enVert[1]{D(\hat{\tau})}_{\ell_\infty}\enVert[1]{z(\hat{\tau})}_{\ell_\infty}\notag\\
&\leq
\sup_{\tau\in T}\enVert[1]{\Theta(\tau)}_{\ell_\infty}\sup_{\tau\in T}\enVert[2]{\del[1]{\Sigma(\tau)-\frac{1}{n}X'(\tau)X(\tau)}}_{\infty}\enVert[1]{\del[0]{\hat{\alpha}-\alpha_0}}_{\ell_1}\notag\\
&+
\sup_{\tau\in T}\enVert[1]{\Theta(\tau)}_{\ell_\infty}\enVert[2]{\frac{1}{n}X'(\hat{\tau})U}_{\ell_\infty}+\lambda \sup_{\tau\in T}\enVert[1]{\Theta(\tau)}_{\ell_\infty}\max_{1\leq j\leq m}\enVert[1]{X^{(j)}}_n\label{alphainf1}
\end{align}
where we have also used $\enVert[1]{z(\hat{\tau})}_{\ell_\infty}\leq 1$. Next, note that $\sup_{\tau\in T}\enVert[1]{\Theta(\tau)}_{\ell_\infty}$ is bounded by assumption. Furthermore, by Lemma \ref{Xu}, $\enVert[2]{\frac{1}{n}X'(\hat{\tau})U}_{\ell_\infty}=O_p\del[2]{\sqrt{\frac{\log(m)}{n}}}$ while $\max_{1\leq j\leq m}\enVert[1]{X^{(j)}}_n=O_p(1)$ by Lemma \ref{scaling}. Finally, it follows by the arguments on page A6 and the last inequality before Appendix B in \cite{lee2012lasso} that $\sup_{\tau\in T}\enVert[2]{\del[1]{\Sigma(\tau)-\frac{1}{n}X'(\tau)X(\tau)}}_{\infty}=O_p\del[2]{\sqrt{\frac{\log(mn)}{n}}}$ while $\enVert[1]{\hat{\alpha}-\alpha_0}_{\ell_1}=O_p\del[2]{s\sqrt{\frac{\log(m)}{n}}}$ by Theorem 2 in the same reference.  Using this in (\ref{alphainf1}) yields, with $\lambda = O\del[1]{\sqrt{\log(m)/n}}$,  
\begin{align*}
\enVert[1]{\hat{\alpha}-\alpha_0}_{\ell_\infty}
=
O_p\del[3]{\sqrt{\frac{\log(m)}{n}}\del[3]{s\sqrt{\frac{\log(mn)}{n}}+2}}
=
O_p\del[3]{\sqrt{\frac{\log(m)}{n}}}
\end{align*}
as $s\sqrt{\frac{\log(mn)}{n}}\to 0$.
\end{proof}

\begin{proof}[Proof of Theorem \ref{thm2}]
First, since $\hat{\alpha}=(\hat{\beta}', \hat{\delta}')'$ satisfies the Karush-Kuhn-Tucker conditions for a minimum, one has
\begin{align*}
-\frac{1}{n}X'(\hat{\tau})\del[1]{Y-X(\hat{\tau})\hat{\alpha}}+\lambda D(\hat{\tau})z(\hat{\tau})
=
0
\end{align*}
where $\enVert[0]{z(\hat{\tau})}_{\ell_\infty}\leq 1$ and$z(\hat{\tau})_j=sign(\hat{\alpha}_j)$ if $\hat{\alpha}_j\neq 0$. This can be rewritten as
\begin{align*}
-\frac{1}{n}X'(\hat{\tau})\del[1]{X(\tau_0)\alpha_0-X(\hat{\tau})\hat{\alpha}}
=
\frac{1}{n}X'(\hat{\tau})U-\lambda D(\hat{\tau})z(\hat{\tau})
\end{align*}
which is equivalent to
\begin{align*}
\frac{1}{n}X'(\hat{\tau})X(\hat{\tau})\del[1]{\hat{\alpha}-\alpha_0}-\frac{1}{n}X'(\hat{\tau})\del[1]{X(\tau_0)-X(\hat{\tau})}\alpha_0
=
\frac{1}{n}X'(\hat{\tau})U-\lambda D(\hat{\tau})z(\hat{\tau}).
\end{align*}
The above display can be rewritten as
\begin{align*}
\Sigma(\tau_0)\del[1]{\hat{\alpha}-\alpha_0}-\frac{1}{n}X'(\hat{\tau})\del[1]{X(\tau_0)-X(\hat{\tau})}\alpha_0
=
\del[1]{\Sigma(\tau_0)-\frac{1}{n}X'(\hat{\tau})X(\hat{\tau})}(\hat{\alpha}-\alpha_0)+\frac{1}{n}X'(\hat{\tau})U-\lambda D(\hat{\tau})z(\hat{\tau}).
\end{align*}
Next $\Theta(\tau_0)=\Sigma(\tau_0)^{-1}$ exists under Assumption 1 by the discussion after Assumption 2 as $\Sigma$ is assumed to exist. In fact, $\kappa=\kappa(s,5, S)>0$ where $S=\cbr{|\tau-\tau_0|\leq \eta_0}$ and $\eta_0=n^{-1}\vee K_1\sqrt{s\lambda}$ \footnote{Here $K_1=\sqrt{{7C_1C_2}}$ where $C_2$ is the constants proven to exist in Lemma \ref{scaling} in the appendix ensuring that $\sup_{\tau\in T}\max_{1\leq j\leq 2m}\enVert[1]{X^{(j)}(\tau)}_n\leq C_2$ with arbitrarily large probability (more precisely, for any $\epsilon>0$ there exists a $C_2$ such that $\sup_{\tau\in T}\max_{1\leq j\leq 2m}\enVert[1]{X^{(j)}(\tau)}_n\leq C_2$ with probability at least $1-\epsilon$).} is satisfied with probability tending to one as is needed in order to invoke Theorem 3 of \cite{lee2012lasso} below (it is even satisfied when $S$ is replaced by $T$). Thus, one may rewrite the above display as
\begin{align*}
\hat{\alpha}-\alpha_0
&=
\Theta(\tau_0)\frac{1}{n}X'(\hat{\tau})\del[1]{X(\tau_0)-X(\hat{\tau})}\alpha_0+\Theta(\tau_0)\del[1]{\Sigma(\tau_0)-\frac{1}{n}X'(\hat{\tau})X(\hat{\tau})}(\hat{\alpha}-\alpha_0)\\
&+\Theta(\tau_0)\frac{1}{n}X'(\hat{\tau})U-\lambda \Theta(\tau_0)D(\hat{\tau})z(\hat{\tau})
\end{align*}
such that arguments similar to those leading to (\ref{alphainf1}) yield
\begin{eqnarray}
\enVert[1]{\hat{\alpha}-\alpha_0}_{\ell_\infty}
&\leq & 
\enVert[1]{\Theta(\tau_0)}_{\ell_\infty}\enVert[2]{\frac{1}{n}X'(\hat{\tau})\del[1]{X(\tau_0)-X(\hat{\tau})}\alpha_0}_{\ell_\infty} \nonumber \\
&+&\enVert[1]{\Theta(\tau_0)}_{\ell_\infty}\enVert[2]{\del[1]{\Sigma(\tau_0)-\frac{1}{n}X'(\hat{\tau})X(\hat{\tau})}}_{\infty}\enVert[1]{\hat{\alpha}-\alpha_0}_{\ell_1} \nonumber \\
&+& \enVert[1]{\Theta(\tau_0)}_{\ell_\infty}\enVert[2]{\frac{1}{n}X'(\hat{\tau})U}_{\ell_\infty}+\lambda \enVert[1]{\Theta(\tau_0)}_{\ell_\infty}\max_{1\leq j\leq n}\enVert[1]{X^{(j)}}_n
\label{feqn}
\end{eqnarray}
where we used that $\enVert[1]{z(\hat{\tau})}_{\ell_\infty}\leq 1$. First, note that $\enVert[1]{\Theta(\tau_0)}_{\ell_\infty}$ is bounded by assumption. Next, denoting by $Z(\tau_0)$ and $Z(\hat{\tau})$ the last $m$ columns of $X(\tau_0)$ and $X(\hat{\tau})$, respectively, one has
\begin{align}
\enVert[2]{\frac{1}{n}X'(\hat{\tau})\del[1]{X(\tau_0)-X(\hat{\tau})}\alpha_0}_{\ell_\infty}
=
\enVert[2]{\frac{1}{n}X'(\hat{\tau})\del[1]{Z(\tau_0)-Z(\hat{\tau})}\delta_0}_{\ell_\infty}\label{rhs}
\end{align}
By Theorem 3 in \cite{lee2012lasso} one has $\envert[0]{\hat{\tau}-\tau_0}=O_p\del[1]{s\frac{\log(m)}{n}}$ such that the probability of $\mathcal{A}=\cbr{\envert[0]{\hat{\tau}-\tau_0}\leq K s\frac{\log(m)}{n}}$ can be made arbitrarily large by choosing $K>0$ sufficiently large. Thus, on $\mathcal{A}$, 
\begin{align*}
\enVert[2]{\frac{1}{n}X'(\hat{\tau})\del[1]{Z(\tau_0)-Z(\hat{\tau}))}\delta_0}_{\ell_\infty}
&\leq
\sup_{1\leq j,k\leq m}\frac{1}{n}\sum_{i=1}^n\envert[2]{X_i^{(j)}X_i^{(k)}}\envert[1]{1_{\cbr{Q_i<\tau_0}}-1_{\cbr{Q_i<\hat{\tau}}}}\enVert[1]{\delta_0}_{\ell_1}\\
&\leq
KC_1s\envert[1]{J(\delta_0)}\frac{\log(m)}{n}
\end{align*}
by Assumptions 1 and 4. As we have assumed that $s|J(\delta_0)|\log(m)^{1/2}/\sqrt{n}\to 0$, we have in particular that

\begin{align}
\enVert[2]{\frac{1}{n}X'(\hat{\tau})\del[1]{X(\tau_0)-X(\hat{\tau})}\alpha_0}_{\ell_\infty}
=
O_p\del[3]{\sqrt{\frac{\log(m)}{n}}}.\label{firstterm}.
\end{align}
Next, note that
\begin{align*}
\enVert[2]{\del[1]{\Sigma(\tau_0)-\frac{1}{n}X'(\hat{\tau})X(\hat{\tau})}}_{\infty}
\leq
\enVert[2]{\del[1]{\Sigma(\tau_0)-\frac{1}{n}X'(\tau_0)X(\tau_0)}}_{\infty}
+
\enVert[2]{\frac{1}{n}\del[1]{X'(\tau_0)X(\tau_0)-X'(\hat{\tau})X(\hat{\tau})}}_{\infty}
\end{align*}
First, by the subgaussianity of the covariates and the error terms Corollary 5.14 in \cite{vershynin2010introduction} and a union bound yield that\footnote{Alternatively, the arguments on pages A4-A6 in \cite{lee2012lasso} yield a uniform (in $\tau$) upper bound on $\enVert[2]{\del[1]{\Sigma(\tau)-\frac{1}{n}X'(\tau)X(\tau)}}_{\infty}$ of the order $O_p\del[2]{\sqrt{\frac{\log(mn)}{n}}}$ which could also be used resulting in only slightly worse rates.} $\enVert[2]{\del[1]{\Sigma(\tau_0)-\frac{1}{n}X'(\tau_0)X(\tau_0)}}_{\infty}=O_p\del[2]{\sqrt{\frac{\log(m)}{n}}}$. Next, by arguments similar to the ones leading to (\ref{firstterm}), one also has
\begin{align*}
\enVert[2]{\frac{1}{n}\del[1]{X'(\tau_0)X(\tau_0)-X'(\hat{\tau})X(\hat{\tau})}}_{\infty}
\leq
\sup_{1\leq j,k\leq m}\frac{1}{n}\sum_{i=1}^n\envert[2]{X_i^{(j)}X_i^{(k)}}\envert[1]{1_{\cbr{Q_i<\tau_0}}-1_{\cbr{Q_i<\hat{\tau}}}}
\leq
K s\frac{\log(m)}{n}
\end{align*}
on $\mathcal{A}$ by Assumption 4. Therefore, as $s\log(m)^{1/2}/\sqrt{n}\to 0$ (implied by our assumption $s|J(\delta_0)|\log(m)^{1/2}/\sqrt{n}\to 0$), we conclude that
\begin{align}
\enVert[2]{\del[1]{\Sigma(\tau_0)-\frac{1}{n}X'(\hat{\tau})X(\hat{\tau})}}_{\infty}
=
O_p\del[3]{\sqrt{\frac{\log(m)}{n}}}\label{secterm}
\end{align}
Furthermore, by Lemma \ref{Xu}, $\enVert[2]{\frac{1}{n}X'(\hat{\tau})U}_{\ell_\infty}=O_p\del[2]{\sqrt{\frac{\log(m)}{n}}}$ and $\enVert[1]{\hat{\alpha}-\alpha_0}_{\ell_1}=O_p\del[2]{s\sqrt{\frac{\log(m)}{n}}}$ by Theorem 3 in \cite{lee2012lasso}. Finally, $\max_{1\leq j\leq m}\enVert[1]{X^{(j)}}_n=O_p(1)$ by Lemma \ref{scaling} which in conjunction with (\ref{firstterm}) and (\ref{secterm}) yields in (\ref{feqn})
\begin{align*}
\enVert[1]{\hat{\alpha}-\alpha_0}_{\ell_\infty}
=O_p\del[3]{\sqrt{\frac{\log(m)}{n}}}
\end{align*}
where have again used that $s\log(m)^{1/2}/\sqrt{n}\to 0$.

\end{proof}

\begin{proof}[Proof of Lemma \ref{ThetaBound1}]
First, note that
\begin{align*}
\Sigma(\tau)= \begin{pmatrix}
  \Sigma & \tau\Sigma \\
  \tau\Sigma & \tau\Sigma \\
 \end{pmatrix}
\end{align*}
such that by the formula for the inverse of a partitioned matrix with $\Theta=\Sigma^{-1}$
\begin{align}
\Theta(\tau)
=
\Sigma^{-1}(\tau)
=
\begin{pmatrix}
 \frac{1}{1-\tau} \Sigma^{-1} & \frac{-1}{1-\tau}\Sigma^{-1} \\
  \frac{-1}{1-\tau}\Sigma^{-1} & \frac{\tau}{\tau(\tau-1)}\Sigma^{-1} \\
 \end{pmatrix}
=
\frac{1}{1-\tau}\begin{pmatrix}
  1 & -1 \\
  -1& \frac{1}{\tau} \\
 \end{pmatrix}
\otimes \Theta \label{Thetat}.
\end{align}
Thus, it suffices to bound $\enVert[1]{\Sigma^{-1}}_{\ell_\infty}$. To this end, note that $\Sigma=(1-\rho)I+\rho\iota\iota'$ where $\iota$ is a $m\times 1$ vector of ones. Thus, by the Sherman-Morrison-Woodbury formula, $\Sigma^{-1}$ exists and equals
\begin{align*}
\Theta
=
\Sigma^{-1}
=
\frac{1}{1-\rho}\del[3]{I-\frac{\rho\iota\iota'}{1-\rho+\rho m}}
\end{align*}
which implies that (using $\rho/(1-\rho+\rho m)\leq 1$)
\begin{align}
\enVert{\Theta}_{\ell_\infty}
=
\frac{1}{1-\rho}\del[3]{1-\frac{\rho}{1-\rho+\rho m}+\frac{\rho(m-1)}{1-\rho+\rho m}}
=
\frac{1}{1-\rho}\del[3]{\frac{1-3\rho+2m\rho}{1-\rho+m\rho}}
\leq
\frac{2}{1-\rho}\label{Thetalinf}.
\end{align}
Thus, combining (\ref{Thetat}) and (\ref{Thetalinf}) yields the first claim of the lemma. The second claim follows trivially from the first.
\end{proof}

\begin{proof}[Proof of Theorem \ref{thmthresh}]
We consider the zero and non-zero coefficients separately and show that both groups will be classified correctly. Note that by Theorems \ref{thm1} and \ref{thm2} for every $\epsilon>0$ there exists a $C>0$ such that $\enVert[0]{\hat{\alpha}-\alpha}\leq C\lambda$ on a set $\mathcal{D}$ with probability at least $1-\epsilon$. The following arguments all take place on this set. Consider the truly zero coefficients first. To this end, let  $j \in J(\alpha_0)^c$ and note that
\begin{align*}
\max_{j\in J(\alpha_0)^c}|\hat{\alpha}_j|
\leq 
C\lambda
<
2C\lambda
=
H
\end{align*}
such that $\tilde{\alpha}=0$ by the definition of the thresholded scaled Lasso.

Next, consider the non-zero coefficients. To this end, let  $j \in J(\alpha_0)$ and note that
\begin{align*}
|\hat{\alpha}_j|
\geq 
\min_{j\in J(\alpha_0)}|\alpha_j|-|\hat{\alpha}_j-\alpha_{j0}  |
\geq
3C\lambda-C\lambda
=
2C\lambda = H
\end{align*}
such that $|\tilde{\alpha}|=|\hat{\alpha}|\neq 0$ by the definition of the thresholded scaled Lasso and the assumption that $\min_{j\in J(\alpha_0)}|\alpha_j| > 3 C \lambda$
\end{proof}

\begin{proof}[Proof of Theorem \ref{thmbreak}]
Proceeds exactly as the proof of Theorem \ref{thmthresh}.
\end{proof}

\bibliographystyle{chicagoa}
\bibliography{references}

\end{document}